\documentclass{article}

\title{A Uniform Improvement of the Benjamini-Hochberg Procedure via e-Closure}

\author{Jelle Goeman\footnote{j.j.goeman@lumc.nl}}

\date{Leiden University Medical Center, The Netherlands}

\usepackage{natbib}
\usepackage{amsmath, amssymb, amsthm, cleveref}
\usepackage{times, diagbox}
\usepackage[margin=3.5cm]{geometry}
\usepackage{tikz}
\usepackage{thmtools, thm-restate}
\usepackage[ruled, vlined]{algorithm2e}
\usepackage{algorithmic}
\usepackage{booktabs}
\usepackage{array}
\usepackage{tabularx}
\usepackage{makecell}
\declaretheorem[name=Lemma]{lemma}

\declaretheorem[name=Theorem, sibling=lemma]{theorem}
\declaretheorem[name=Remark, sibling=lemma]{remark}
\declaretheorem[name=Assumption, sibling=lemma]{assumption}

\renewcommand{\a}{\mathbf{a}}
\renewcommand{\b}{\mathbf{b}}
\newcommand{\e}{\mathbf{e}}

\newcommand{\p}{\mathbf{p}}
\newcommand{\q}{\mathbf{q}}
\renewcommand{\k}{\mathbf{k}}
\renewcommand{\r}{\mathbf{r}}
\renewcommand{\v}{\mathbf{v}}
\newcommand{\E}{\mathrm{E}}
\newcommand{\EE}{\mathbf{E}}
\newcommand{\I}{\mathbf{I}}
\renewcommand{\P}{\mathrm{P}}
\newcommand{\R}{\mathbf{R}}
\newcommand{\RR}{\boldsymbol{\mathcal{R}}}
\newcommand{\cBH}{closed BH}

\SetCommentSty{mycommfont}
\SetKwComment{tcp}{}{}

\begin{document}

\maketitle

\begin{abstract}
This paper presents \emph{closed BH}, a uniform improvement of the False Discovery Rate controlling method of Benjamini and Hochberg (BH). Closed BH is valid under the same assumption of Positive Regression Dependency on a Subset (PRDS) as BH, but also under an alternative and weaker minimal sufficient condition. As a uniform improvement, closed BH never rejects fewer hypotheses than BH, but it may reject quite a few more. An increase in power is observed especially when the number of false null hypotheses is large. The novel method is constructed using the e-Closure principle, a recently derived general principle for multiple testing. The method is implemented in the \emph{eClosure} package in R.
\end{abstract}

\section{Introduction}

In their seminal paper, \citet{BenjaminiHochberg1995} introduced the False Discovery Rate (FDR), the proportion of false discoveries among the discoveries, and advocated its use as the error rate of choice in moderate- and large-scale multiple testing. In the same paper they introduced a simple and elegant method controlling FDR, which became known as the Benjamini-Hochberg (BH) procedure. In their original paper, \citet{BenjaminiHochberg1995} proved that BH controls FDR when p-values are independent. \citet{BenjaminiYekutieli2001} extended the applicability of BH by proving the validity of the same procedure under a weaker condition of positive correlation among the p-values, the assumption of Positive Regression Dependency on a Subset (PRDS). Now more than thirty years later, FDR control has established itself firmly as the dominant error rate for multiple testing in many scientific fields, and BH remains the primary workhorse for controlling FDR.

However, there are some indications that BH is not optimal as a method for controlling FDR. Designed for control at level $\alpha$, it controls at the more stringent level $\pi_0\alpha$ instead, where $\pi_0$ is the proportion of true null hypotheses \citep{BenjaminiYekutieli2001}. This suggests that a power gain might be achievable. Much research effort has been made to try to repair the gap between $\alpha$ and $\pi_0\alpha$. The most common approach is to find a suitable estimate $\hat\pi_0$ of $\pi_0$ and plug this into BH \citep{Storey2002,StoreyTaylorSiegmund2004, BenjaminiKriegerYekutieli2006, BlanchardRoquain2008, Sarkar2008,BlanchardRoquain2009, GavrilovBenjaminiSarkar2009, FinnerDickhausRoters2009, HeesenJanssen2016, MacDonaldLiangJanssen2019, Gao2023, GaoRoquain2025, IgnatiadisWangRamdas2026}. While such methods often reject more than BH, they all require either a conservative correction factor to $\alpha$ or must allow the possibility that $\hat\pi_0>1$. As a result, these improvements over BH are not uniform, i.e., these methods may sometimes reject fewer hypotheses than BH. The net result is a gain in power relative to BH when $\pi_0$ is small (and power is plentiful), but a loss when $\pi_0$ is large (and power is low), a trade-off few researchers are willing to make. Moreover, none of these methods has proven validity under PRDS. \cite{FithianLei2022} uniformly improved BH under positive dependence, but their method requires knowledge of the precise form of the dependency, which the PRDS assumption does not offer. Only one method, Minimally Adaptive BH \citep[MABH;][]{SolariGoeman2017}, uniformly improves BH and is valid under PRDS, but the power gain of that method relative to BH is negligible unless the number of hypotheses is very small.

This paper uses a different approach to improve BH, exploiting the recently proposed e-Closure principle \citep{XuSolariFischerDeHeideRamdasGoeman2025}. This principle says that every FDR-controlling procedure is a special case of a general e-closed procedure based on e-values, a recently proposed expectation-based alternative to p-values \citep{VovkWang2021, GrunwaldDeHeideKoolen2024, RamdasWang2025}. An e-closed procedure is a generalization of Closed Testing \citep{MarcusPeritzGabriel1976} to expectation-based error rates. E-closed procedures are based on an e-collection, which is an analogue of the suite of local tests that defines a closed testing procedure. By writing known procedures in the form of e-closed procedures, they can sometimes be uniformly improved, as \citet{XuSolariFischerDeHeideRamdasGoeman2025} have demonstrated for the FDR-controlling procedures of \citet{BenjaminiYekutieli2001}, \citet{Su2018} and \citet{WangRamdas2022}. \citet{XuSolariFischerDeHeideRamdasGoeman2025} did not improve BH itself beyond MABH, and left the question of whether or not it could be improved by e-Closure as an open problem. This paper solves that open problem by constructing an appropriate e-closed procedure, \emph{closed BH}, which uniformly improves BH, and is valid under PRDS. 

Like adaptive BH methods, closed BH exploits the gap between $\pi_0\alpha$ and $\alpha$ and rejects more hypotheses than BH especially when $\pi_0$ is small. However, unlike plug-in methods, the improvement of closed BH over BH is uniform, implying that the researcher choosing closed BH never loses out on any rejections BH would have made. Rather than estimating and plugging in $\hat\pi_0$, the method operates by exploiting the event that some p-values are much smaller than needed for rejection by BH; the existence of such p-values makes rejection of `later' hypotheses with larger p-values easier. Valid under PRDS, closed BH actually only requires a more lenient minimal sufficient condition, which will be derived later.

The paper is organized as follows. We present a class of ``BH-like'' methods in \Cref{sec: BHlike description}, and an e-collection in \Cref{sec: e-collection}. \Cref{sec: FDR control} combines the two, showing through e-Closure that all BH-like methods control FDR. BH-like methods are a family of methods that depend on a matrix of parameters, which includes MABH as a special case. \Cref{sec: closed BH} defines the closed BH method by arguing for a single canonical choice for these parameters. We discuss simultaneous FDR control and algorithms in \Cref{sec: simultaneous,sec: algorithms}. Simulations in \Cref{sec: simulations} illustrate the gains of the new method relative to BH. Proofs are found in Appendix \ref{sec: proofs}.

As a preview, \Cref{tab: bh_real_data_discoveries} gives the number of hypotheses rejected for several data sets used for illustration in well-known publications on multiple testing, the same that were used by \citet{XuSolariFischerDeHeideRamdasGoeman2025}. As a uniform improvement, closed BH will never reject fewer hypotheses than BH, and may reject the same number, but it may reject quite a few more, as illustrated in the table.

\begin{table}[!ht]
\caption{
    Number of discoveries by the NH and MABH procedures vs.\ the new closed BH procedure across multiple simulated and real datasets and FDR thresholds $\alpha$. The selection of datasets follows \citet{XuSolariFischerDeHeideRamdasGoeman2025}.
    Bold indicates where closed BH makes more discoveries than BH. Sources: APSAC: \citet{BenjaminiHochberg1995}; NAEP: \citet{BenjaminiHochberg2000}; PADJUST:  \texttt{p.adjust} R function; PVALUES \texttt{fdrtool} R package; VANDEVIJVER: \citet{GoemanSolari2014}; GOLUB: \citet{EfronHastie2021}.}
\label{tab: bh_real_data_discoveries}
    \centering
\begin{tabular}{lccccccc}
\toprule
Dataset & \makecell{\# hypotheses} & \multicolumn{3}{c}{$\alpha = 0.05$} & \multicolumn{3}{c}{$\alpha = 0.1$} \\ 
\cmidrule(lr){3-5} \cmidrule(lr){6-8}
 &  & BH & MABH & \cBH & BH & MABH & \cBH   \\
\midrule
APSAC & 15 & 4 & 4 & 4 & 9 & 9 & 9  \\
NAEP & 34 & 11 & 11 & \textbf{12} & 12 & 12 & \textbf{15}  \\
PADJUST & 50 & 20 & 21 & \textbf{21} & 21 & 21 & \textbf{22}  \\
PVALUES & 4,289 & 767 & 767 & \textbf{801} & 1139 & 1139 &  \textbf{1214}  \\
VANDEVIJVER & 4,919 & 1340 & 1340 & \textbf{1412} & 1728 & 1728 &  \textbf{1856}  \\
GOLUB & 7,128 & 1249 & 1249 & \textbf{1275} & 1605 & 1605 & \textbf{1688}  \\
\bottomrule
\end{tabular}
\end{table}

\section{BH-like methods} \label{sec: BHlike description}

We start by defining a family of ``BH-like'' methods that depend on a matrix of parameters $r_{k,s}$. Closed BH and MABH will emerge in \Cref{sec: closed BH} 
as members of this family for specific choices of $r_{k,s}$.

This paper uses the same notational conventions used by \citet{XuSolariFischerDeHeideRamdasGoeman2025} that lowercase is for both scalars and vectors, capital letters indicate sets, and caligraphic, e.g., $\mathcal{R}$, collections of sets. Boldface indicates random variables.

Suppose we have null hypotheses $H_1, \ldots, H_m$. Let $\p = (\p_1, \ldots, \p_m)$ be p-values for these null hypotheses. We write $\p_{(1)}\leq \ldots \leq \p_{(m)}$ for the ordered p-values, and let $\p_{(0)} = 0$ for convenience. Let $\I_r$ denote the index set of the $r$ smallest p-values, with ties broken according to original index.

Let integer \emph{horizon} values $r_{k,s}$ be defined for $0 \leq k \leq m$ and $1 \leq s \leq m$. The name of horizon will be explained in detail in \Cref{sec: FDR control}; roughly, the horizon $r_{k,s}$ indicates the maximal size of the rejection set that $\p_{(k)}$ can achieve in step $s$ of the algorithm. Initialize, for all $k+s\leq m+1$,
\begin{equation} \label{eq: ass r initialization}
r_{k,s} = \begin{cases}
    m & \textrm{if $k=1$ and $s=m$;} \\
    k & \textrm{otherwise.} 
\end{cases}
\end{equation}
This leaves $r_{k,s}$ still undefined for $k+s > m+1$. These remaining values of $r_{k,s}$ can be chosen freely subject to two constraints. We need, for all $k$ and $s$,
\begin{equation} \label{eq: ass r greater k}
k \leq r_{k,s} \leq m,
\end{equation}
and, whenever $k > 0$,
\begin{equation} \label{eq: ass r_k+1 upper bound}
r_{k-1,s} \leq r_{k,s} \leq b_{k,s},   
\end{equation}
where $b_{k,s}$ depends on $r_{k-1,s}$ through
\[
b_{k,s} = \begin{cases}
m \wedge \frac{ (m-s)(k+s-m-1)r_{k-1,s}}{(k+s-m)(m-s) - r_{k-1,s}} & \textrm{if $(k+s-m)(m-s)-r_{k-1,s} >0$;} \\
m & \textrm{otherwise.}
\end{cases}
\]
We see that $r_{k,s}$ is non-decreasing in $k$ for each $s$ from $r_{0,s}=0$ to $r_{m,s}=m$, but that \eqref{eq: ass r_k+1 upper bound} constrains the growth not to be ``too fast'' (which will be necessary for \Cref{lem: aks nondecreasing} later). We have no choice for $s=m$, where we must take $r_{k,m} = m$ for $k>0$ because of the initialization $r_{1,m}=m$. Constraints \eqref{eq: ass r greater k} and \eqref{eq: ass r_k+1 upper bound} are automatically fulfilled for the $r_{k,s}$ defined in \eqref{eq: ass r initialization} for $k+s\leq m+1$, so they only need to be verified for the free part. An example choice of $r_{k,s}$ is given in \Cref{tab: example r}.

\begin{table}[!ht] 
\centering
\begin{tabular}{c|ccccccc}
  \diagbox{$k$}{$s$} & 1 & 2 & 3 & 4 & 5 & 6 & 7 \\
  \hline
  0 & 0 & 0 & 0 & 0 & 0 & 0 & 0 \\
  1 & 1 & 1 & 1 & 1 & 1 & 1 & 7 \\
  2 & 2 & 2 & 2 & 2 & 2 & 2 & 7 \\
  3 & 3 & 3 & 3 & 3 & 3 & \textbf{4} & 7 \\
  4 & 4 & 4 & 4 & 4 & \textbf{6} & \textbf{7} & 7 \\
  5 & 5 & 5 & 5 & \textbf{6} & \textbf{7} & \textbf{7} & 7 \\
  6 & 6 & 6 & \textbf{6} & \textbf{7} & \textbf{7} & \textbf{7} & 7 \\
  7 & 7 & 7 & 7 & 7 & 7 & 7 & 7 \\
\end{tabular}
\caption{Example horizon values $r_{k,s}$ for $m=7$. The values of $r_{k,s}$ for which there is some freedom of choice are in bold.} \label{tab: example r}
\end{table}

We will now verify that Constraints \eqref{eq: ass r greater k} and \eqref{eq: ass r_k+1 upper bound} are not (self-)contradictory. \Cref{lem: bk viable 1} proves that Constraint \eqref{eq: ass r_k+1 upper bound} always allows at least one value for $r_{k,s}$.

\begin{restatable}{lemma}{bkviablefirst}
\label{lem: bk viable 1}
$b_{k,s} \geq r_{k-1,s}$.
\end{restatable}

\Cref{lem: bk viable 2} shows that the bound $k \leq r_{k,s} \leq b_{k,s}$ implied by combining Constraints \eqref{eq: ass r greater k} and \eqref{eq: ass r_k+1 upper bound} also always allows at least one value for $r_{k,s}$.

\begin{restatable}{lemma}{bkviablesecond}
\label{lem: bk viable 2}
$b_{k,s} \geq k$
\end{restatable}

Based on the chosen $r_{k,s}$, we define \emph{threshold} values
\begin{equation} \label{def: a}
a_{k,s} = \begin{cases}
    \frac{(k+s-m)r_{k,s}\alpha}{(r_{k,s}+s-m)s} & \textrm{if $r_{k,s} > k$}; \\
    \frac{k\alpha}{s} & \textrm{otherwise.}
\end{cases}
\end{equation}
In step $s$ of the algorithm, each $\p_{(k)}$ will be compared to the threshold value $a_{k,s}$ to see if it contributes. We have $a_{k,s} > 0$ whenever $k>0$ and $\alpha>0$, since $r_{k,s} = k$ when $k+s\leq m$. Note that $(k+s-m)r_{k,s}/(r_{k,s}+s-m)s$ actually reduces to $k/s$ naturally when $r_{k,s}=k$, so the split in the definition of $a_{k,s}$ is only needed to avoid 0/0 issues when $k+s=m$. Remark also that the value of $r_{k,s}$ is immaterial for $a_{k,m}$, since $a_{k,m} = k\alpha/m$ always. This is the rationale behind the seemingly incongruent initialization choice $r_{1,m} = m$ in \eqref{eq: ass r initialization}. For the example of the $r_{k,s}$ of \Cref{tab: example r}, the threshold values for $a_{k,s}$ are displayed (up to a factor $\alpha$) in \Cref{tab: example a}.

\begin{table}[!ht] 
\centering
\begin{tabular}{c|ccccccc}
  \diagbox{$k$}{$s$} & 1 & 2 & 3 & 4 & 5 & 6 & 7 \\
  \hline
  0 & 0 & 0 & 0 & 0 & 0 & 0 & 0 \\
  1 & 1 & 1/2 & 1/3 & 1/4   & 1/5   & 1/6   & 1/7 \\
  2 & 2 & 1   & 2/3 & 1/2   & 2/5   & 1/3   & 2/7 \\
  3 & 3 & 3/2 & 1   & 3/4   & 3/5   & 4/9   & 3/7 \\
  4 & 4 & 2   & 4/3 & 1     & 3/5   & 7/12  & 4/7 \\
  5 & 5 & 5/2 & 5/3 & 1     & 21/25 & 7/9   & 5/7 \\
  6 & 6 & 3   & 2   & 21/16 & 28/25 & 35/36 & 6/7 \\
  7 & 7 & 7/2 & 7/3 & 7/4   & 7/5   & 7/6   & 1 \\
\end{tabular}
\caption{The threshold values $a_{k,s}/\alpha$ corresponding to the choice of $r_{k,s}$ of \Cref{tab: example r}.} \label{tab: example a}
\end{table}

Next, we compare $\p_{(0)}, \ldots, \p_{(m)}$ to these thresholds to calculate \emph{thresholded horizon} values, defined as
\[
\r_{k,s} = r_{k,s} 1\{\p_{(k)} \leq a_{k,s}\}.
\]
Suppose that $\p_1 = \p_2 = 2\alpha/7$; $\p_3 = \p_4 = 7\alpha/12$; $\p_5 = \alpha$; $\p_6 = 15\alpha/8$, and $\p_7=2\alpha$ have been observed and that $r_{k,s}$ had been chosen (independently of these p-values) according to the example of \Cref{tab: example r}. Then the thresholded horizon values are given in \Cref{tab: example thresholded r}.

\begin{table}[!ht] 
\centering
\begin{tabular}{c|ccccccc}
  \diagbox{$k$}{$s$} & 1 & 2 & 3 & 4 & 5 & 6 & 7 \\
  \hline
  0 & 0 & 0 & 0 & 0 & 0 & 0 & 0 \\
  1 & 1 & 1 & 1 & 0 & 0 & 0 & 0 \\
  2 & 2 & 2 & 2 & 2 & 2 & 2 & \textbf{7} \\
  3 & 3 & 3 & 3 & 3 & 3 & 0 & 0 \\
  4 & 4 & 4 & 4 & 4 & \textbf{6} & \textbf{7} & 0 \\
  5 & 5 & 5 & 5 & \textbf{6} & 0 & 0 & 0 \\
  6 & \textbf{6} & \textbf{6} & \textbf{6} & 0 & 0 & 0 & 0 \\
  7 & 7 & 7 & 7 & 0 & 0 & 0 & 0 \\
\end{tabular}
\caption{The thresholded horizon values $\r_{k,s}$ corresponding to \Cref{tab: example r} for $\p_1 = \p_2 = 2\alpha/7$; $\p_3 = \p_4 = 7\alpha/12$; $\p_5 = \alpha$; $\p_6 = 15\alpha/8$, and $\p_7=2\alpha$. The boldface values are the values of $\r_{\k_s,s}$, the locations driving the cumulative maximum at $\r$. The $\k_s$ are defined in \eqref{def: ks}.} \label{tab: example thresholded r}
\end{table}

We reject $k$ hypotheses if, for every step
$s$ from 1 to  $m$, some index $k' \leq k$ has a horizon reaching to at least $k$. From $\r_{k,s}$ we therefore define cumulative maxima over $k$ for each $s$, i.e.,
\[
\v_{k,s} = \max_{0\leq k'\leq k} \r_{k',s},
\]
and look for the largest $k$ such that the cumulative maxima at $k$ are consistently (over $s$) above this $k$:
\begin{equation} \label{def: r}
\r = \max\{0 \leq k \leq m\colon \v_{k,s} \geq k \textrm{\ for all $1\leq s\leq m$}\}.
\end{equation}
We can check that $\r$ is always defined since $\v_{0,s}=0$ always. 

The BH-like procedure now rejects the $\r$ hypotheses with smallest p-values, i.e., the hypotheses with indices in $\I_\r$. In the example, we see that $\r$ takes the value 6. We can easily compare this to the number of rejections of BH for the same p-values. Since the $a_{k,m} = k\alpha/m$ are exactly the critical values of BH, the number of rejections of BH is the largest $k$ for which $\r_{k,m} \geq k$, which is 2 in this case. 

In the next sections, we will investigate and motivate this class of methods. In \Cref{sec: e-collection} we will put the BH-like method in the context of the e-Closure principle \citep{XuSolariFischerDeHeideRamdasGoeman2025} by constructing an e-collection for this method, and we will use this principle in \Cref{sec: FDR control} to prove FDR control for the method as sketched here.
During these investigations, the rationale of the procedure sketched above, which may seem a bit peculiar at first sight, will become clearer.

\section{An e-collection} \label{sec: e-collection}

The e-Closure principle \citep{XuSolariFischerDeHeideRamdasGoeman2025} links FDR control closely to the concept of e-values \citep{VovkWang2021, GrunwaldDeHeideKoolen2024, RamdasWang2025}. An e-value for a null hypothesis is a non-negative test statistic whose expectation is at most one under the null hypothesis. E-values can be used for hypothesis testing like p-values, rejecting when the e-value exceeds $1/\alpha$. A useful property of e-values is that the average of e-values is again an e-value.

The e-Closure principle says that to make an FDR-controlling procedure, we should construct an e-collection of $2^m$ e-values $\e_S \geq 0$, for all $S \in 2^{[m]}$, such that
\begin{equation} \label{eq: minimally sufficient}
\E(\e_N) \leq 1,
\end{equation}
where $N \subseteq [m]$ denotes the unknown set of indices of the hypotheses that are true. Here, $[m]$ is a shorthand for $\{1,\ldots, m\}$ and $2^{[m]}$ is the power set of $[m]$. An e-collection is essentially a collection of $2^m$ e-values, one for every intersection hypothesis $H_S = \bigcap_{i\in S} H_i$. However, rather than all e-values for all true $H_S$ having expectation at most 1, we require this only for $\e_N$, noting that $H_N$ is always a true hypothesis. 

Based on such an e-collection,  \citet{XuSolariFischerDeHeideRamdasGoeman2025} show that we can make an FDR-controlling procedure using \Cref{thm: e-Closure}. In fact, \citet{XuSolariFischerDeHeideRamdasGoeman2025} show that \Cref{thm: e-Closure} is necessary for FDR control, in the sense that all FDR-controlling procedures can be reconstructed, and often improved, by appealing to the e-Closure Principle.

\begin{theorem}[e-Closure] \label{thm: e-Closure}
If, for some e-collection $\EE = (\e_S)_{S \in 2^{[m]}}$, we have 
\begin{equation} \label{eq: e-closure}
\e_S \geq \frac{|\R \cap S|}{(|\R|\vee 1)\alpha}\ \textrm{\ for all $S \in 2^{[m]}$},
\end{equation}
then $\R$ controls FDR.
\end{theorem}

We will now make an e-collection for the BH-like procedure. First, define
\begin{equation} \label{def: ks}
\k_s = \max \{0\leq k \leq \r\colon \p_{(k)} \leq a_{k,s}\}.
\end{equation}
Since $r_{k,s}$ is non-decreasing in $k$ by \eqref{eq: ass r_k+1 upper bound}, $\k_s$ is a largest location $k$ at which $\r_{k,s}$ attains the cumulative maximum in the definition of $\v_{\r,s}$. Note that we have $\v_{\k_s,s} \geq \r$ by definition of $\r$. The value of $\k_s$ are always defined since $\p_{(0)} = 0 = a_{0,s}$, so the set in \eqref{def: ks} is never empty. The values of $\k_s$ are illustrated for the running example in \Cref{tab: example thresholded r}. The boldface numbers there are the $\r_{\k_s,s} = \v_{\k_s,s}$. The $\k_s$ themselves are the corresponding row locations. 

Using these $\k_s$, we define the following e-collection. Let $\e_\emptyset = 1$ and, for $S\neq \emptyset$, $\e_S=0$ if $\k_{|S|}=0$, and otherwise
\begin{equation} \label{def: e-collection}
\e_S = \frac{\sum_{i\in S} 1\{p_i \leq \a_{|S|}\}}{|S|\a_{|S|}}, 
\end{equation}
where $\a_{|S|} = a_{\k_{|S|},|S|}$, which is random through $\k_{|S|}$. The $\e_S$ averages e-values constructed per p-value with index in $S$, and the threshold applied for each p-value depends on $|S|$. 

The minimal sufficient condition that our proof of FDR control of closed BH will require is that \eqref{eq: minimally sufficient} holds for this e-collection, i.e.,
\begin{assumption}[Minimal Sufficient Condition] \label{assumption}
\[
\EE \bigg(\frac1{|N|} \sum_{i\in N} \frac{1\{\p_i \leq \a_{|N|}\}}{\a_{|N|}}\bigg) \leq 1
\]
\end{assumption}

We will now show that this minimal sufficient condition is implied by PRDS. The argument relies on a very useful result of \citet{BlanchardRoquain2008}, repeated here as \Cref{lem: BR}. According to that theorem, it suffices to show that $\a_{n}$, with $n=|N|$, is coordinatewise non-increasing in $\p$ in order to show that our e-collection \eqref{def: e-collection} is valid under PRDS.

\begin{theorem}[Blanchard \& Roquain] \label{lem: BR}
Suppose $\p_1\ldots, \p_m$ are valid, i.e., 
\begin{equation} \label{eq: ass p valid}
\P(\p_i \leq t)\leq t\  \textrm{for all $t \in [0,1]$ and all $i \in N$;}
\end{equation}
and satisfy PRDS, i.e.,
\begin{equation} \label{eq: ass PRDS}
\P(\p \in D \mid \p_i \leq x)\ \textrm{is non-decreasing in $x$ for all decreasing sets $D$ and all $i \in N$.}
\end{equation}
Then, for every $i \in N$, and every coordinatewise non-increasing function $f$, we have
\[
\E \bigg(\frac{1\{\p_i \leq f(\p)\}}{f(\p)} \bigg) \leq 1.
\]
\end{theorem}
    
We show the necessary monotonicity in \Cref{lem: ks nonincreasing,lem: aks nondecreasing}. First, \Cref{lem: ks nonincreasing} shows that $\k_s$ increases as the p-values decrease. This is quite intuitive from their construction. As the p-values get smaller the $\r_{k,s}$ flip from 0 to $r_{k,s}$, allowing the boldface ``front'' in \Cref{tab: example thresholded r} to move down.

\begin{restatable}{lemma}{ksmonotone} \label{lem: ks nonincreasing}
$\k_s$ is coordinatewise non-increasing in $\p$.
\end{restatable}

\Cref{lem: aks nondecreasing} shows that monotonicity of $\a_{s}$ follows from monotonicity of $\k_s$. It relies heavily on Constraint \eqref{eq: ass r_k+1 upper bound}, and the form of $b_{k,s}$ was chosen there as the maximal bound needed for the proof of \Cref{lem: aks nondecreasing} to work.
 
\begin{restatable}{lemma}{aksmonotone} \label{lem: aks nondecreasing}
$a_{k,s}$ is non-decreasing in $k$.
\end{restatable}

Taking \Cref{lem: ks nonincreasing,lem: aks nondecreasing} and \Cref{lem: BR} together, we have the e-collection we need.

\begin{theorem} \label{thm: e-collection}
If \eqref{eq: ass p valid} and \eqref{eq: ass PRDS} hold, then $\EE = (\e_S)_{S \in 2^{[m]}}$ defined in \eqref{def: e-collection} is an e-collection.
\end{theorem}

\begin{proof}
By \Cref{lem: ks nonincreasing}, $\k_{|N|}$ is coordinatewise non-increasing in $\p$ and by \Cref{lem: aks nondecreasing} $a_{k,|N|}$ is non-decreasing in $k$. Therefore $\a_{|N|} = a_{\k_{|N|}, |N|}$ is coordinatewise non-increasing in $\p$. By \Cref{lem: BR}, 
\[
\E(\e_N) = \E \bigg(\frac1{|N|} \sum_{i \in N} \frac{1\big\{\p_i \leq \a_{|N|}\big\}}{\a_{|N|}} \bigg) = 
\frac1{|N|} \sum_{i \in N} \E \bigg(\frac{1\big\{\p_i \leq \a_{|N|}\big\}}{\a_{|N|}} \bigg) \leq 1.
\]
\end{proof}

\begin{remark}
\citet{BenjaminiYekutieli2001} showed the validity of BH under PRDS. \citet{Sarkar2008PRDS} and \citet{FithianLei2022} showed that PRDS is actually quite a restrictive condition. Many important distributions for which BH appears to be valid in practice do not satisfy PRDS. Indeed, PRDS is sufficient but not necessary for FDR control of BH. From the proof of \citet{BlanchardRoquain2008}, we see that a minimal sufficient condition for the validity of BH for FDR control at level $\pi_0\alpha$ is that
\begin{equation} \label{eq: min suf BH}
\EE \bigg(\frac1{|N|} \sum_{i\in N} \frac{1\{p_i \leq \b\alpha/m\}}{\b\alpha/m}\bigg) \leq 1,
\end{equation}
where $\b$ is the number of BH rejections, defined formally in \eqref{eq: BH} below. This condition is much less restrictive than PRDS, which is a condition on all combinations of all null p-values for all increasing sets. Instead, we have a single condition on the expectation of a single random variable. Similarly, PRDS is sufficient but not necessary for closed BH, and also \Cref{assumption} requires only a single, and very comparable expectation to be checked.
\end{remark}

\section{FDR control} \label{sec: FDR control}

We will now show that the FDR-controlling method implied by the e-collection of \Cref{sec: e-collection} is equal to the BH-like method described in \Cref{sec: BHlike description}. We will do this by showing that \eqref{eq: e-closure} holds for $\R = \I_\r$, with $\r$ defined in \eqref{def: r}, for the e-collection $\EE$ as defined in \eqref{def: e-collection}.

The main workhorse for this is \Cref{lem: minimal e}. This lemma shows that if we can make sure that the $k$th largest e-value is at least $1/a_{k,s}$, which happens in our e-collection if the $k$th smallest p-value is at most $a_{k,s}$, then this facilitates the necessary condition \eqref{eq: e-closure} for all $S$ with $|S|=s$ not just for $\R=\I_k$, but of all $\I_k, \ldots, \I_{r_{k,s}}$. This lemma motivates the name of $r_{k,s}$ as a horizon, since it allows $\p_{(k)}$ to ``look ahead'' to larger potential rejection sets than $\I_k$ up to the horizon $\I_{r_{k,s}}$. It also explains the precise form of the definition of $a_{k,s}$ in \eqref{def: a}, which was chosen as the maximal value that makes \Cref{lem: minimal e} work. 

\begin{restatable}{lemma}{minimale} \label{lem: minimal e}
Suppose $e_1, \ldots, e_m \geq 0$ are given and let $0 < k \leq m$ exist such that for all $1 \leq i \leq k$,
\[
e_i \geq 1/a_{k,s}.
\]
Then for every $S \subseteq [m]$ with $|S|=s$ and for every $R = [j]$ with $k \leq j \leq r_{k,s}$, we have
\[
\frac1{|S|} \sum_{i \in S} e_i \geq \frac{|R \cap S|}{|R|\alpha}.
\]
\end{restatable}

The FDR control condition \eqref{eq: e-closure} of e-Closure for $\I_\r$ follows almost immediately from \Cref{lem: minimal e}. The argument is that, by construction, the conditions for \Cref{lem: minimal e} hold for $k=\k_s$, and that, also by construction, $\k_s \leq \r \leq r_{\k_s,s}$. 

\begin{restatable}{lemma}{eSOK} \label{thm: eS OK}
Let $\e_S$ be defined in \eqref{def: e-collection} and $\r$ in \eqref{def: r}. Then for all $S \in 2^{[m]}$ we have
\[
\e_S \geq \frac{|\I_\r \cap S|}{(\r\vee 1)\alpha}.
\] 
\end{restatable}

FDR control for $\I_\r$ is now immediate from the e-Closure principle.

\begin{theorem} \label{thm: FDR control}
Let $\r$ be as defined in \eqref{def: r}. If \Cref{assumption} holds, then the method that rejects the hypotheses with indices in $\I_\r$ has FDR control at level $\alpha$, i.e.,
\[
\E \bigg( \frac{|\I_\r \cap N|}{\r\vee 1} \bigg) \leq \alpha.
\]
\end{theorem}
    
\begin{proof}
Combine \Cref{thm: e-Closure} with \Cref{thm: eS OK}.
\end{proof}

The proof of validity under PRDS of the method follows by combining \Cref{thm: FDR control} with \Cref{thm: e-collection}. If PRDS does not hold, we may not have \Cref{assumption}. However, if we know by how much \Cref{assumption} is violated, we know exactly what level of FDR control is retained. \Cref{thm: FDR gamma} makes this precise. This theorem can be used to show robustness: violations of PRDS that lead to small $\gamma$ lead to small excess FDR.

\begin{theorem} \label{thm: FDR gamma}
Let $\r$ be as defined in \eqref{def: r}. If
\[
\EE \bigg(\frac1{|N|} \sum_{i\in N} \frac{1\{\p_i \leq \a_{|N|}\}}{\a_{|N|}}\bigg) \leq \gamma,
\]
then the method that rejects the hypotheses with indices in $\I_\r$ has FDR control at level $\gamma\alpha$, i.e.,
\[
\E \bigg( \frac{|\I_\r \cap N|}{\r\vee 1} \bigg) \leq \gamma\alpha.
\]
\end{theorem}
    
\begin{proof}
Let $\tilde \e_S = \e_S/\gamma$ for all $S \in 2^{[m]}$. Then $(\tilde\e_S)_{S \in 2^{[m]}}$ is an e-collection. By \Cref{thm: eS OK},  for all $S \in 2^{[m]}$,
\[
\tilde\e_S  = \frac{\e_S}{\gamma} \geq \frac{|\I_\r \cap S|}{(\r\vee 1)\gamma\alpha}.
\] 
The result now follows from \Cref{thm: e-Closure}.
\end{proof}

\section{Closed BH} \label{sec: closed BH}

We have developed not a single FDR-controlling procedure under PRDS, but a whole family of such methods. This family is parameterized by the choice of $r_{k,s}$. Though subject to some constraints, the number of possible configurations for $r_{k,s}$ explodes very quickly as $m$ grows. In this section, we will investigate properties of methods resulting from some of these choices and define the closed BH method.

The choice of $r_{k,s}$ represents a true trade-off. A choice $r'_{k,s} > r_{k,s}$ means that $\r'_{k,s}$ is larger than $\r_{k,s}$ if the former is not zero, but, because $a'_{k,s} < a_{k,s}$, $\r'_{k,s}$ is more likely to be zero than $\r_{k,s}$. Therefore, one choice of $r_{k,s}$ does not generally dominate another. We can think of larger $r_{k,s}$ as ``ambitious,'' high-risk high-gain, and smaller $r_{k,s}$ as ``cautious''. However, the net effect of any choice of $r_{k,s}$ depends on the full configuration  of all the other $r_{k,s}$ and the resulting landscape is complex. 

Since our stated purpose is to improve BH, a frame of reference will be the performance of BH-like methods relative to BH and to MABH, which is known to improve BH. BH rejects a number
\begin{equation} \label{eq: BH}
\b = \max\{0\leq k \leq m\colon p_{(k)} \leq k\alpha/m\}
\end{equation}
of hypotheses, while MABH rejects
\begin{equation} \label{def: MABH}
\b' = 1\{\b>0\} \cdot \max\{0\leq k \leq m\colon p_{(k)} \leq k\alpha/(m-1)\}. 
\end{equation}

We will first consider the most cautious choice of $r_{k,s}$. By \eqref{eq: ass r greater k} the minimal choice for $r_{k,s}$ is
\begin{equation}
   \check r_{k,s} = \begin{cases}
        m & \textrm{if $s=m$ and $k \neq 0$;}\\
        k & \textrm{otherwise.}
    \end{cases}
\end{equation}
Let $\check\r$ be the number of rejections resulting from this choice of $r_{k,s}$. This choice exactly reconstructs MABH, as \Cref{lem: MABH} claims. It follows that MABH is the most cautious among all BH-like methods.

\begin{restatable}{lemma}{lemMABH} \label{lem: MABH}
    $\check\r = \b'$
\end{restatable}

Although the improvement of MABH over BH is uniform, the improvement is also tiny, especially if $m$ is large \citep{SolariGoeman2017}. We have, trivially,
\begin{equation} \label{eq: MABH limited}
\b' \leq \b_{m\alpha/(m-1)},
\end{equation}
where $\b_{\alpha'}$ is the number of BH-rejections at level $\alpha'$. This relationship is immediate from \eqref{def: MABH}. It follows that the gain of MABH relative to BH vanishes as $m$ grows, and the improvement of MABH over BH has rightfully been described as ``tiny'' \citep{SolariGoeman2017}. Still, it turns out that we cannot uniformly improve upon MABH within the class of BH-like methods, as \Cref{lem: MABH admissible} claims. We call one BH-like method rejecting $\r'$ a uniform improvement of another BH-like method rejecting $\r$ if, for all $\alpha>0$, (1.) $\r' \geq \r$ for all $\p$, and (2.) $\r' > \r$ for at least one choice of $\p$.

\begin{restatable}{proposition}{MABHadmissible} \label{lem: MABH admissible}
No BH-like method uniformly improves MABH.
\end{restatable}

By \Cref{lem: MABH admissible} we cannot hope to improve upon MABH uniformly with the theory we have developed in this paper. However, we can improve upon BH in alternative and less cautious ways than MABH does. Indeed, MABH is not the only BH-like method that uniformly improves BH. \Cref{lem: improve BH} characterizes a necessary condition for BH-like procedures uniformly to improve BH. They do so if they are not too ambitious. 

\begin{restatable}{proposition}{criterionBH} \label{lem: improve BH}
We have $\r \geq \b$ for all $\alpha>0$ if and only if 
\[
r_{k,s} \leq \frac{ks}{m-k}
\]
whenever $k+s > m+1$, $s<m$ and $k < m$.
\end{restatable}

The proof of \Cref{lem: improve BH} is essentially the observation that $r_{k,s} \leq ks/(m-k)$ is equivalent to $a_{k,s} \geq k\alpha/m$. 

\begin{remark}
It is perhaps counterintuitive that according to the definition, BH is not itself a BH-like method. However, BH could be constructed with the recipe of \Cref{sec: BHlike description} if the requirement that $r_{1,m} = m$ would be relaxed, thus allowing the choice that $r_{k,s}=k$ for all $k,s$. This choice recovers BH, as is easily checked. This choice is suboptimal, however, since $a_{k,m}$ does not depend on $r_{k,m}$ at all, so we may take $r_{k,m}$ maximal without paying any price. Since a larger $r_{k,s}$ for the same $a_{k,s}$ is never worse, $r_{k,m} = k$ is suboptimal.
\end{remark}

\Cref{lem: improve BH} suggests that alternatives to MABH exist for uniformly improving BH. Because of the underwhelming performance of the maximally cautious MABH, we define \emph{closed BH} as the maximally ambitious method improving BH uniformly, i.e., the BH-like method with maximal $r_{k,s}$ while still satisfying the condition of \Cref{lem: improve BH}. It is defined recursively as the BH-like method with the choice
\[
\bar r_{k,s} = \min\Big( b_{k,s}, \Big\lfloor \frac{ks}{m-k}\Big\rfloor \Big) ,
\]
for $k+s> m+1$ and $s<m$ from the initialization \eqref{eq: ass r initialization}, where we remember that $r_{k,m}=m$ is forced for $k>0$. Let $\bar\r$ be the resulting number of rejections. The example of \Cref{tab: example r} corresponds to the $r_{k,s}$ of closed BH for $m=7$. It illustrates that closed BH is indeed a uniform improvement of BH, since, with the p-values of \Cref{tab: example thresholded r}, the method rejects more than BH.

It is interesting to contrast closed BH to MABH. We see from the proof of \Cref{lem: MABH admissible} that a method trying to be at least as good as MABH everywhere must have $r_{k,s}=k$ like MABH. In contrast, \Cref{lem: improve BH} gives much more room for $r_{k,s}$ if the goal is to improve BH. In the example of \Cref{tab: example r}, the only choice of $r_{k,s}$ that was restricted by the constraint $r_{k,s} \leq ks/(m-k)$ is at $k=3$, $s=6$; all others were already constrained by $b_{k,s}$. The price to pay, in terms of lower critical values $a_{k,s}$ compared to MABH, is limited however, since by \eqref{eq: MABH limited} the critical values of closed BH are at least those of BH, so $(m-1)/m$ times those of MABH. This is a small price to pay for a large increase in $r_{k,s}$.

However, the constraint due to \Cref{lem: improve BH} does bind closed BH closely to MABH when the number of rejections of the latter method is small. We have that if $k < \sqrt{m}$, then $ks/(m-k) < k+1$. Therefore, $\bar r_{k,s} = k$ if $k < \sqrt{m}$, implying that closed BH and MABH have identical $r_{k,s}$ for $k < \sqrt{m}$. The consequence of this observation is made clear in \Cref{lem: sqrt m}.

\begin{restatable}{proposition}{lemsqrtm} \label{lem: sqrt m}
If $\check{\r} < \sqrt{m}$, then $\bar \r = \check{\r}$.
\end{restatable}

\Cref{lem: sqrt m} limits the potential gain of closed BH relative to MABH, but also the potential loss. In practice, as we will see in \Cref{tab: bh_real_data_discoveries} and \Cref{sec: simulations}, closed BH tends to reject more than MABH when the number of rejections of the latter method is large. This suggests that, although closed BH is not a uniform improvement of MABH in theory, it is a uniform improvement in practice. MABH outperforms closed BH primarily in the unusual situation that most of the p-values are nearly equal. A corollary to \Cref{lem: sqrt m} is that $\b=0$ implies $\bar\r=0$. Therefore, closed BH has the same implied global test as BH. In fact, $\b=0$ implies $\r=0$ for all BH-like methods, since $\v_{k,m}=0$ for all $k$ if $\b=0$.


\section{Simultaneity} \label{sec: simultaneous}

So far, we have used the e-Closure principle to show FDR control of a single set $\R$, as in the classical definition of FDR control of \citet{BenjaminiHochberg1995}. However, \citet{XuSolariFischerDeHeideRamdasGoeman2025} showed that any procedure that controls FDR also controls simultaneous FDR. Let, for some e-collection $(\e_S)_{S \in 2^{[m]}}$, a collection of rejection sets be defined as
\begin{equation} \label{eq: eClosure simultaneous}
\boldsymbol{\mathcal{R}} = \Big\{R \in 2^{[m]}\colon \e_S \geq \frac{|R \cap S|}{(|R| \vee 1)\alpha} \textrm{\ for all $S \in 2^{[m]}$}\Big\}.
\end{equation}
Then we have simultaneous FDR control over all $R \in \boldsymbol{\mathcal{R}}$, i.e., we have
\[
\E \Big( \max_{R \in \boldsymbol{\mathcal{R}}} \frac{|R \cap N|}{|R|\vee 1} \Big) \leq \alpha.
\]
Because the max is inside the expectation, simultaneous FDR control allows a researcher to choose any set $\R$ post hoc from $\RR$, while still controlling FDR for this set.

Such simultaneity may also occur with BH-like methods. To see an example, consider the BH-like method illustrated in \Cref{tab: example thresholded r}. Here, we found $\r=6$, implying that $\R=[6]$ may be rejected. However also $\R' = \{1,2,3,4,5,7\}$ can be rejected. To see why, note that $\p_6$ and $\p_7$ pass and fail exactly the same thresholds $\a_s$, $s \in [m]$, and that membership of in $\RR$ is determined solely by passing and failing of thresholds, not by the ranking of p-values. It follows that $\R'$ is an equivalent rejected set to $[6]$. In addition, we always, trivially, have $\emptyset \in \boldsymbol{\mathcal{R}}$.

Although some simultaneity exists for BH-like methods, as we have just seen, it is actually highly limited. The precise extent of the simultaneity of BH-like methods is characterized by \Cref{thm: hardly simultaneous}, which says that all possible cases of simultaneity exchange one or more rejected hypotheses from $\I_\r$ for the same number of other hypotheses with similar p-values, just as happened in the example. In particular, non-trivial simultaneity is restricted to sets of the same size. Applying \Cref{thm: hardly simultaneous} to the example of \Cref{tab: example thresholded r}, we see that $\boldsymbol{\mathcal{R}} = \{\emptyset, [6], \{1,2,3,4,5,7\}\}$. 

\begin{restatable}{theorem}{hardlysimultaneous} \label{thm: hardly simultaneous}
Let $\boldsymbol{\mathcal{R}}$ be defined in \eqref{eq: eClosure simultaneous}, based on the e-collection with $\e_S$ defined in \eqref{def: e-collection}. We have $R \in \boldsymbol{\mathcal{R}}$ if and only if $R = \emptyset$ or 
\[
|R| = \r,
\]
and, for every $s$,
\[
|\{i \in R\colon \p_i \leq \a_{s}\}| \geq \k_s.
\]
\end{restatable}

Note that $R= \I_\r$ fulfills the criteria of the theorem by definition. For some configurations of $r_{k,s}$, the second condition of \Cref{thm: hardly simultaneous} can be fulfilled only for $R = \I_\r$. MABH is an example of this, and \citet{XuSolariFischerDeHeideRamdasGoeman2025} already proved that $\boldsymbol{\mathcal{R}} = \{\emptyset, \I_\r\}$ for MABH.

The simultaneity offered by \Cref{thm: hardly simultaneous} is much more limited than the simultaneity that has been achieved by improving some other methods through the e-Closure Principle. For the closed versions of the methods of \citet{BenjaminiYekutieli2001}, \citet{Su2018} and \citet{WangRamdas2022}, \citet{XuSolariFischerDeHeideRamdasGoeman2025} showed in particular that any additional rejections in the closed method relative to the original are optional: the researcher may choose post hoc whether to use the new, larger rejected set or to retain the original. In contrast, closed BH replaces the original BH-rejected set by a different, potentially larger random set. The researcher using closed BH may not decide post hoc to take the BH-rejected set after all. Closed BH is therefore a uniform improvement of BH in the classical sense that requires $\R'\supseteq \R$ and $\R'\supset \R$ for some $\p$, since $\bar \r \geq \b$, but not according to the simultaneity-focused definition of \citet{XuSolariFischerDeHeideRamdasGoeman2025} that requires $\RR'\supseteq \RR$ and $\RR'\supset \RR$ for some $\p$, since generally $\I_\b \notin \RR$. Moreover, the collection $\RR$ for the methods in \citet{XuSolariFischerDeHeideRamdasGoeman2025} typically contains sets of different sizes, while the simultaneity of closed BH is strictly limited to sets of the same size. We can say that the simultaneity offered by \Cref{thm: hardly simultaneous} is essentially negligible.

\section{Algorithms} \label{sec: algorithms}

Algorithm \ref{alg:cBH} below implements the general BH-like procedure, calculating the number of hypotheses rejected. A naive implementation following the construction of \Cref{sec: BHlike description} would take $O(m^2)$ time and memory. The implementation of Algorithm \ref{alg:cBH} reduces memory use to $O(m)$ and computation time to about $O(m\log(m) + \r^2)$ after sorting.

The algorithm works as follows. It starts by considering $[m]$ as candidate values for $\r$. These are stored in a doubly linked list for $O(1)$ removal of elements during traversal, while always keeping track of the maximum. Iterating through all $s \in [m]$, it discards candidate values that do not meet the criterion $\v_{k,s} \geq k$. Iteration on $s$ proceeds in breadth-first order, i.e., roughly considering integers close to $m/2$, $3m/4$, $m/4$, $7m/8$, $5m/8$, \ldots, until all $s\in[m]$ are enumerated. Exploring more diverse values of $s$ in this way makes it likely that the candidate set shrinks faster. Within each $s$, the algorithm iterates on $k$, starting from $k=m-s+1$. This late start is allowed because, for fixed $k$, $\p_{(k)} \leq a_{k,s}$ for $k+s=m+1$, implies 
\[
\p_{(k)} \leq a_{k,s} = \frac{k\alpha}{s} \leq \frac{k\alpha}{s'} = a_{k,s'}
\]
for all $s'<s$. Therefore, if $k$ fails the criterion at $s'$, it also fails at the boundary $s=m-k+1$. It follows that it is sufficient for the algorithm to check only the region $k+s \geq m+1$. Iteration on $k$ can stop after $k$ reaches the largest remaining value in the candidate set, since $\r \geq k$. Iteration on $k$ can also stop when the cumulative maximum $\v_{k,s}$ exceeds the largest remaining value in the candidate set, since in that case no more values will be removed from the candidate set. The algorithm achieve approximately $O(m\log (m) + \r^2)$, since large values are pruned relatively quickly in about $O(m\log m)$ time because of the breadth-first search. Next, verification requires checking about $m$ values of $s$ and $\r$ values of $k$, but the late start reduces this burden from $O(m\r)$ to $O(\r^2)$. 

\begin{algorithm}[!ht]
\caption{General BH-like procedure} \label{alg:cBH}
\KwIn{$\p_1, \ldots, \p_m$: p-values; $\alpha$: significance level; $\mathtt{next.r}$: horizon function}
\KwOut{Number of rejections}

Sort $\p_{(1)} \leq \cdots \leq \p_{(m)}$\;
Let $\tilde{p}_k = \p_{(k)} / \alpha$ for $k = 1, \ldots, m$\tcp*[r]{embed $\alpha$ into the p-values}
Initialize ordered set $K = [m]$\tcp*[r]{doubly linked list; $O(1)$ max and removal}

\BlankLine
\ForEach(\tcp*[f]{breadth-first order; see text}){$s \in \{1, \ldots, m\}$}{

  $k \gets m - s + 1$\;
  $r \gets r_{k,s}$\tcp*[r]{initialize horizon; $m$ if $s=m$; $k$ otherwise} 
  $v \gets 0$\tcp*[r]{initialize cumulative maximum}

  \BlankLine
  \While{$k \leq \max K$}{
      $a \gets \dfrac{(k + s - m)\, r}{(r + s - m)\, s}$\tcp*[r]{critical value $a = a_{k,s}/\alpha$}

    \BlankLine
    \lIf(\tcp*[f]{cumulative maximum $v = \v_{k,s}$}){$\tilde{p}_k \leq a$}{$v \gets \max(v, r)$}

    \BlankLine
    \If{$v < k$  \textbf{\emph{and}} $k \in K$}{
      Remove $k$ from $K$\tcp*[r]{$k$ fails the criterion; guard against re-removal}
      \lIf{$K = \emptyset$}{\Return $0$}
    }
    \lIf(\tcp*[f]{all remaining $k$ survive this $s$}){$v \geq \max K$}{\textbf{break}}
    \BlankLine
    $k \gets k + 1$\tcp*[r]{advance one step (may pass over removed indices)}
    $r \gets \mathtt{next.r}(r, m, s, k)$\tcp*[r]{incremental horizon update}
  }
}

\Return{$\max K$}\;
\end{algorithm}

Instead of counting rejections at fixed $\alpha$, we may also calculate FDR-adjusted $p$-values. The $i$th FDR-adjusted p-value $\q_i$ is defined as the smallest $\alpha$-level for which $H_i$ is rejected, i.e.,
\[
\q_i = \min\{\alpha'\colon \{i\} \in \R_{\alpha'}\}, 
\]
where we emphasize the dependence of the rejected set $\R$ on $\alpha$. FDR-adjusted p-values are often misunderstood. It is important to realize that FDR-adjusted p-values cannot be interpreted as a property of $H_i$ in isolation, but are always a joint property of the rejected set and the hypothesis \citep[Section 5.4]{GoemanSolari2014}.

The adjusted p-value Algorithm \ref{alg:cBH:adjust}, given in pseudocode in Appendix \ref{sec: alg adjusted}, follows a similar structure to Algorithm \ref{alg:cBH}, but instead of removing candidates, it computes for each $k$ the minimal significance level $\alpha'$ at which $k$ would survive all $s$. At $s$, hypothesis $k$ survives if there exists $k' \leq k$ with $r_{k',s} \geq k$ and $\p_{(k')} \leq a_{k',s}$. The minimal $\alpha'$ for $k$ to survive step $s$ is therefore 
\[
\min_{k' \leq k,\, r_{k',s} \geq k} \frac{\p_{(k')}}{\tilde a_{k',s}}, 
\]
where $\tilde a_{k,s} = a_{k,s}/\alpha$. The $\q_{(k)}$, i.e., the minimal $\alpha'$ for $k$ to survive all steps, is the maximum of this quantity over $s$. Since $r_{k',s}$ is non-decreasing in $k'$, the constraint $r_{k',s} \geq k$ defines a window $[l, k]$ whose left edge $l$ advances as $k$ increases. The minimum of $\p_{(k')}/\tilde a_{k',s}$ over this sliding window is maintained efficiently using a monotone double-ended queue (deque), giving $O(1)$ cost per step. Unlike the candidate-set algorithm, no elements are removed, so a breadth-first ordering of $s$ would give no pruning advantage; instead, $s$ is simply iterated from $m$ down to $1$. The early stopping rules from the candidate-set algorithm also no longer apply, making the worst-case computational complexity $O(m^2)$, although memory use is $O(m)$. Since the procedure returns the largest survivor rather than individual survival decisions, a final pass takes the cumulative minimum from the right: $\q_{(k)} \gets \min_{k' \geq k} \q_{(k')}$.

Like other $\pi_0$-adaptive FDR controlling methods, BH-like methods can reject hypotheses with p-values greater than $\alpha$, and as a consequence adjusted p-values can be smaller than unadjusted. This can be confusing for practitioners. Algorithm \ref{alg:cBH:adjust} therefore has the option to restrict adjusted p-values to be at least as large as unadjusted, which is enabled by default.

Both algorithms are implemented in the \texttt{eClosure} R package, available on cran.

\section{Applications and simulations} \label{sec: simulations}

\subsection{Simulation}

A small simulation experiment was performed to investigate the improvement in power of closed BH versus BH and MABH. Data were generated from an equi-correlated standard normal distribution with correlation $\rho$ and dimension $m=200$. To the first $\pi_1m$ coordinates, with $\pi_1 = 1-\pi_0$, a signal was added of magnitude
\[
\mu = \Phi^{-1}\bigg(1 - \frac{t\alpha(1-\pi_0)}{1-\pi_0\alpha}\bigg) - \Phi^{-1}(1-t),
\]
for a chosen target power of $t$, where $\Phi^{-1}$ is the standard normal quantile function. This size of signal gives an approximately constant average power over $\pi_0$ of magnitude $t$ for BH if $\rho=0$ \citep{StoreyTaylorSiegmund2004} if one-sided tests are used. In the two-sided case similar constant average power is achieved if the value of $\mu$ is taken as the solution of 
\[
1-t = \Phi\bigg(\Phi^{-1}\bigg(1-\frac{t\alpha(1-\pi_0)}{2(1-\pi_0\alpha)}\bigg) - \mu\bigg) - \Phi\bigg(-\Phi^{-1}\bigg(1-\frac{t\alpha(1-\pi_0)}{2(1-\pi_0\alpha)}\bigg) - \mu\bigg),
\]
where $t$ again is the target power. Next, one-sided and two-sided p-values were calculated. The parameters were chosen as $\rho = 0, 0.3, 0.6, 0.9$, $t=0.2, 0.4, 0.6, 0.8$, and $\pi_0 = 0.1, 0.2, \ldots, 0.9$. The value of $\alpha$ was fixed at 0.05. All results were based on $10^5$ simulation runs.

\begin{figure}[!ht]
\centering
\includegraphics[width = 0.95\textwidth]{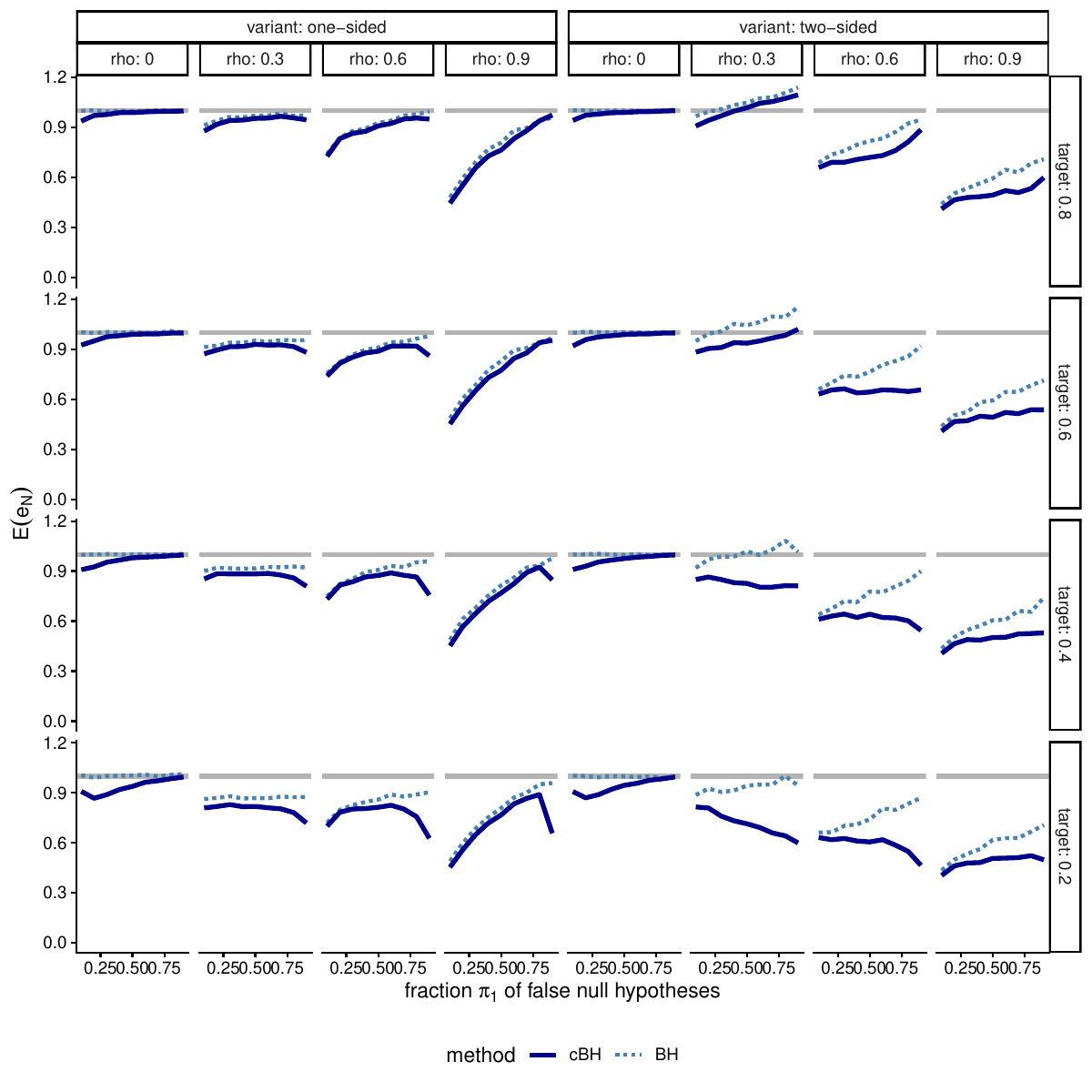}
\caption{$\E(\e_N)$ for closed BH (``cBH'') and BH for the simulation settings of \Cref{sec: simulations}. Here, \emph{target} is the parameter $t$.} \label{fig: msc}
\end{figure}

First, we checked \Cref{assumption} for closed BH and the corresponding assumption \eqref{eq: min suf BH} for BH. The results are given in \Cref{fig: msc}. It is known that PRDS holds for p-values from one-sided normal tests with positive correlations \citep{Sarkar2008PRDS}, but not for the two-sided case \citep{FithianLei2022}. Since PRDS is sufficient for \Cref{assumption}, we see $\E(\e_N) \leq 1$ in \Cref{fig: msc} everywhere for one-sided tests. For the two-sided case, we see some mild violations for small $\rho$, high power $t$ and high $\pi_1$ (low $\pi_0$). From this plot and \Cref{thm: FDR gamma} we might expect some violation of FDR control, but to a maximal level of about 0.055, since we never observe $\gamma > 1.1$. Interestingly closed BH tends to keep $\E(\e_N)$ lower than BH. This may relate to $\k_{|N|}$ depending on all of $\r_{k,s}$ rather than only on $\r_{k, |N|}$.

\begin{figure}[!ht]
\centering
\includegraphics[width = 0.95\textwidth]{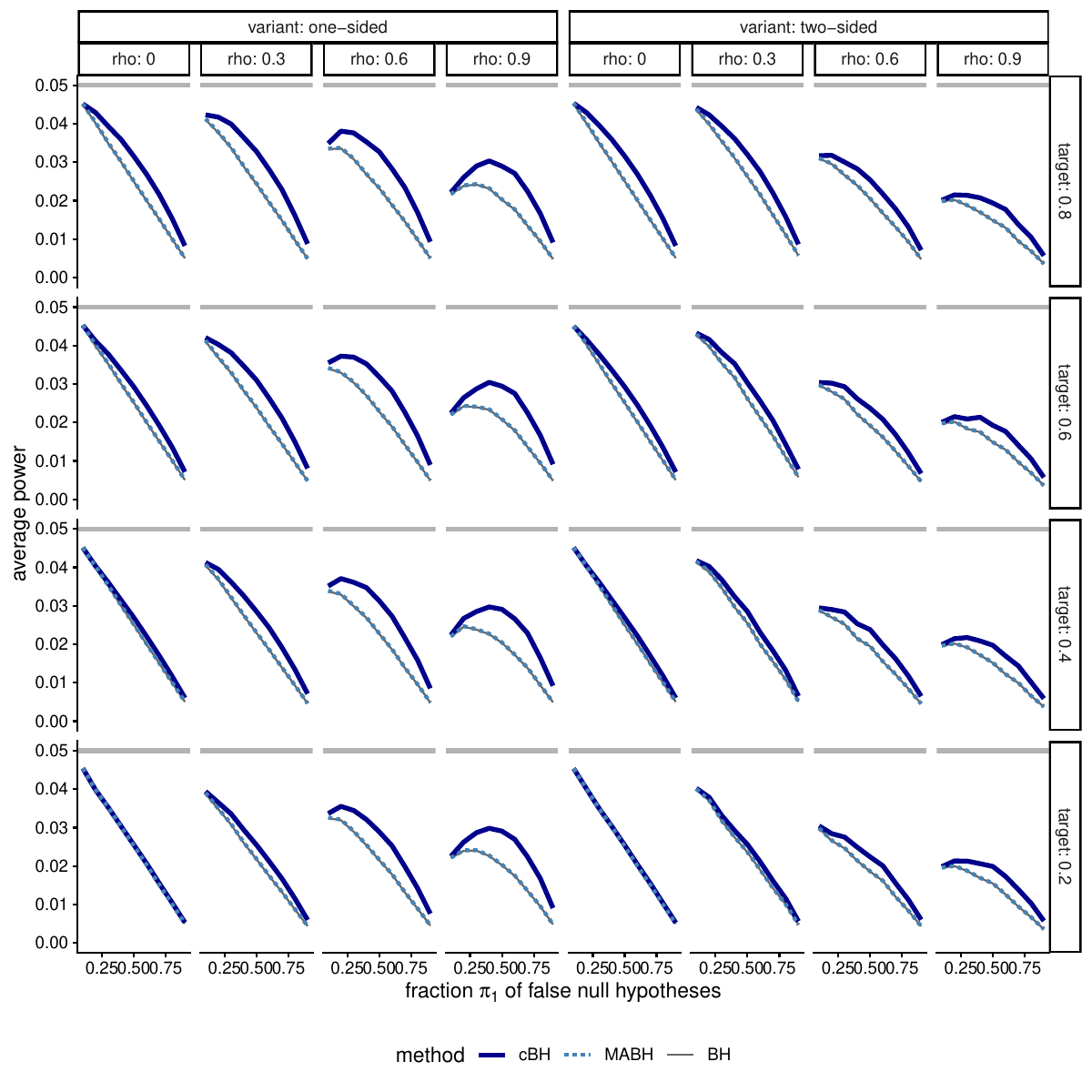}
\caption{Realized FDR of closed BH (``cBH'') versus MABH and BH for the simulation settings of \Cref{sec: simulations}.} \label{fig: fdr}
\end{figure}

\Cref{fig: fdr} gives the realized FDR for closed BH, MABH and BH. This figure confirms that \Cref{assumption} is sufficient but not necessary for FDR control. We see that FDR control is below nominal level for all simulation scenarios with $\E(\e_N)\leq 1$ from \Cref{fig: msc}, but also in the problematic cases. To understand why FDR can still be controlled while $\E(\e_N) > 1$, note that some $e_S$ with $S \neq N$ may take low values due to imperfect power. Such low e-values reduce the number of rejections, resulting in a lower power but also in a lower realized FDR. This is part of the explanation of the robustness of BH, which seems to translate, at least in this example, to closed BH.

Average power is given in \Cref{fig: ap}. We see that MABH is all but indistinguishable from BH, as advertised \citep{SolariGoeman2017}, since $m$ is not small. Closed BH, however, obtains moderate to substantial gains relative to BH/MABH. These gains are greatest when $\pi_0$ is small and when $\rho$ is small, and when $t$ is large, i.e., in situations when the number of rejections of BH/MABH is substantial, as expected from \Cref{lem: sqrt m}. Although closed BH is not a uniform improvement of MABH by \Cref{lem: MABH admissible}, there is never a noticeable power loss against the latter method.

\begin{figure}[!ht]
\centering
\includegraphics[width = 0.95\textwidth]{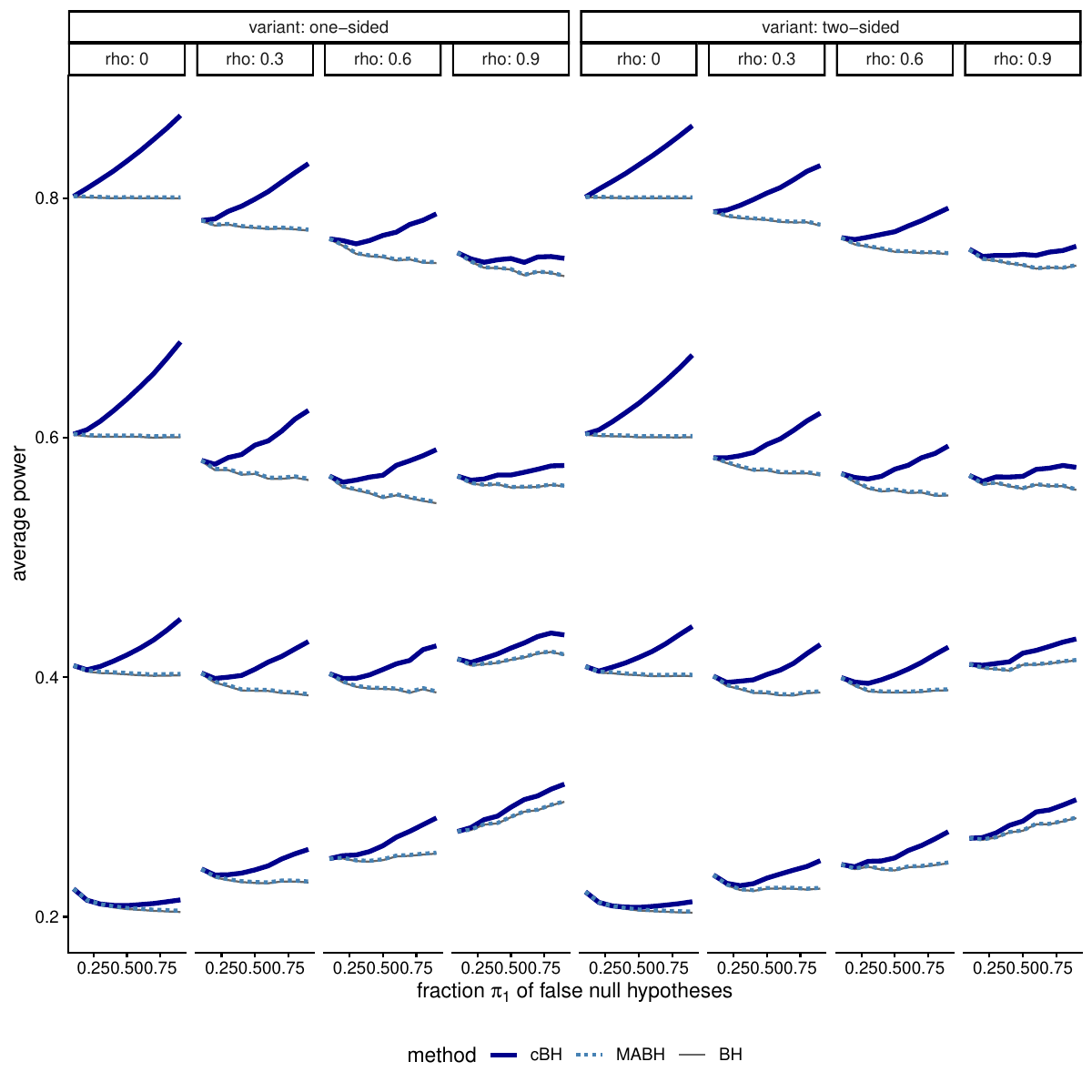}
\caption{Average power of closed BH (``cBH'') versus MABH and BH for the simulation settings of \Cref{sec: simulations}.} \label{fig: ap} 
\end{figure}

\subsection{Applications}

Next, adjusted p-values were calculated for several data sets used for illustration in well-known publications on multiple testing, following \citet{XuSolariFischerDeHeideRamdasGoeman2025}. The number of rejected hypotheses for closed BH versus BH is given in \Cref{tab: bh_real_data_discoveries}. \Cref{fig: adjusted} gives the adjusted p-values of closed BH versus BH. The plot gives both the variant that constrains the adjusted p-values to be never lower than the unadjusted and the unconstrained variant.

\begin{figure}[!ht]
\centering
\includegraphics[width = 0.8\textwidth]{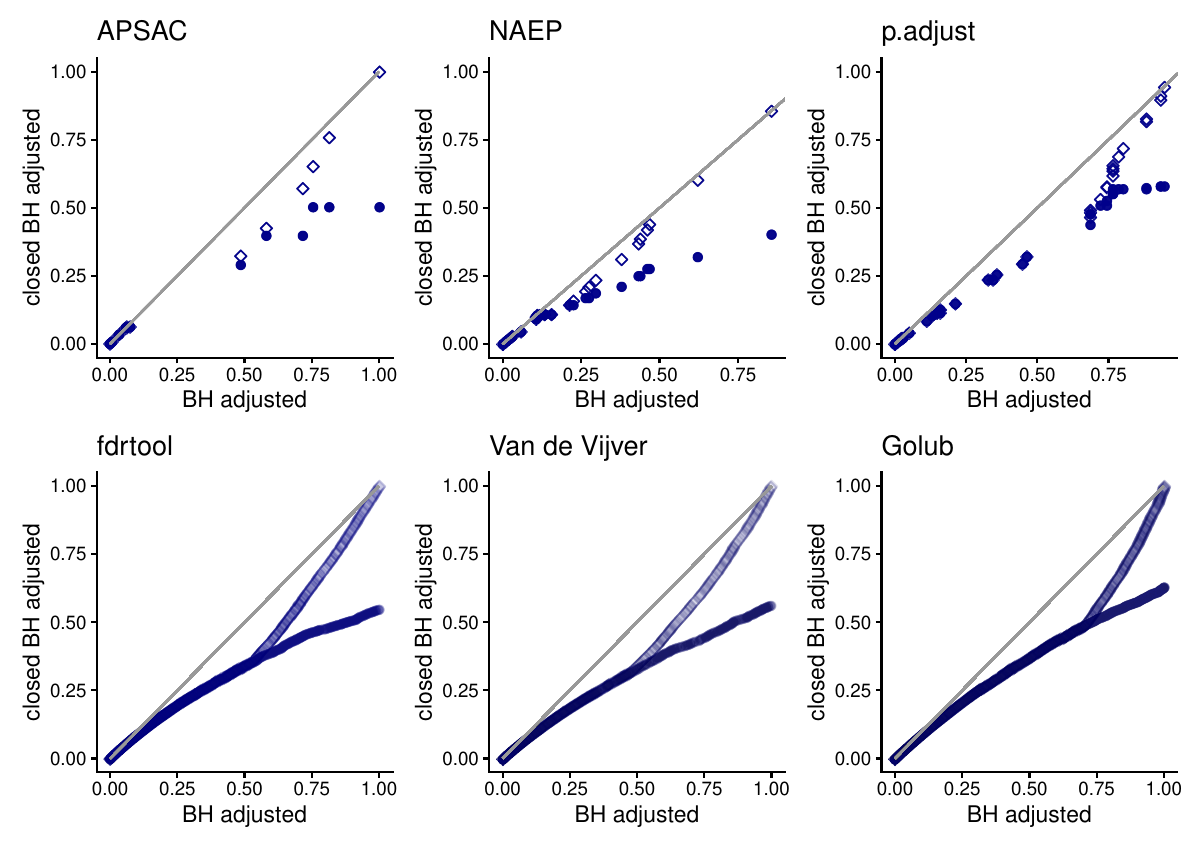}
\caption{Adjusted p-values with closed BH versus BH for the 6 data sets of \Cref{tab: bh_real_data_discoveries}. The open dots (upper curve) represent the adjusted p-values constrained to be not smaller than the unadjusted. The filled dots (lower curve) give the unconstrained adjusted p-values.} \label{fig: adjusted}
\end{figure}

The difference between the closed BH and BH adjusted p-values increases with the p-values. The most noticeable difference is the maximal adjusted p-value, which in BH is always simply the maximal p-value. Closed BH, in contrast, may already reject all hypotheses when the $\alpha$-level is much smaller than that. In contrast, the smallest adjusted p-value of closed BH is always equal to that of BH, since $\b=0$ implies $\bar\r=0$ for all $\alpha$.

\section{Discussion}

This paper derived closed BH, a uniform improvement of the method of Benjamini and Hochberg that has proven validity under PRDS, and under a more lenient minimal sufficient condition implied by PRDS. The power gain relative to BH is appreciable. It can be especially substantial if the proportion $\pi_1$ of false null hypotheses is large. Because the improvement is uniform, a researcher aiming to maximize the number of rejected hypotheses, while controlling FDR at $\alpha$ under PRDS, should always prefer closed BH over BH. The additional cost is only computational. It is true that closed BH offers no meaningful simultaneity, but so does BH, so there is nothing lost there.

The improvement was not derived in the classical way, by applying BH on an adjusted level $\alpha/\hat\pi_0$, for a suitable estimate $\hat\pi_0$ of $\pi_0$. Rather, the derivation proceeded through the e-Closure Principle of \citet{XuSolariFischerDeHeideRamdasGoeman2025}. The improvement mechanism, therefore, is different from that of the plug-in methods. Rather than using the entire p-value distribution, it exploits the event that some p-values are much smaller than needed for rejection to make rejection of later (larger) p-values easier. This has something of the flavor of a step-down method, but without  the $\alpha$-level adjustment that was needed for step-down methods so far \citep{BenjaminiKriegerYekutieli2006, GavrilovBenjaminiSarkar2009}.

Closed BH was shown to be valid under PRDS, but also under a milder minimal sufficient condition. This minimal sufficient condition follows directly from the e-Closure principle, which says that the validity of a method derived from e-Closure rests on the validity of a single e-value $\e_N$ \citep{XuSolariFischerDeHeideRamdasGoeman2025}. If the assumptions of the method are violated and $e_N$ is not an e-value, the expectation of $e_N$ can be used to bound the anti-conservativeness of the method. This property can hopefully be used in the future to extend the scope of closed BH beyond PRDS and to investigate its robustness. Still, even the new condition is only sufficient and not necessary for FDR control, and we have seen situations in which FDR control was achieved even when the minimal sufficient condition was not met. This latter property also holds for BH, which is known to be quite robust against violation of PRDS. It remains to be investigated whether closed BH has a similar robustness.

Unlike other novel FDR-control methods derived from e-Closure so far \citep{XuSolariFischerDeHeideRamdasGoeman2025, LiGoeman2026}, closed BH does not offer much in terms of simultaneity of FDR control. This is a substantial drawback of the method, which apparently emphasizes power over simultaneity. Lack of simultaneity means that known problems surrounding non-simultaneous FDR control with BH remain relevant for closed BH. In particular, researchers may not reduce the set of rejected hypotheses of closed BH in any way while still claiming FDR control \citep{FinnerRoters2001, GoemanSolari2014, EbrahimpoorGoeman2021, KatsevichSabattiBogomolov2023}. Future alternative uniform improvements of BH that emphasize simultaneity rather than power could therefore be more relevant than closed BH for some application areas.

Care must also be taken when interpreting FDR-adjusted p-values. Although BH never gives adjusted p-values that are smaller than their unadjusted counterpart, this property is not a feature of FDR methods in general, and most $\pi_0$-adaptive FDR control methods return adjusted p-values smaller than unadjusted. This can be confusing for researchers, and this confusion stems from an overinterpretation of adjusted p-values \citep[Section 5.4]{GoemanSolari2014}. To avoid some of this confusion, adjusted p-values are capped by default in the software in such a way that they cannot be smaller than unadjusted p-values. 

Closed BH is a member of a whole class of BH-like methods which all control FDR under PRDS. Closed BH was motivated as the most ``ambitious'' member of this class under the contraint that it should uniformly improve BH, i.e., the furthest such uniform improvement from MABH, which is at the other extreme of the same class, and which is known to have low power. It is worthwhile investigating the class as a whole, choosing the parameters $r_{k,s}$ to achieve certain specific goals. For example, if a researcher is not interested in rejecting more than, say, $n$ hypotheses, then capping $r_{k,s}$ at $n$ for $k \leq n$ may gain some power. It is also interesting to investigate whether BH-like methods that do not uniformly improve BH, i.e., which do not fulfil the condition of \Cref{lem: improve BH}, might gain substantial power in a trade-off. It is conceivable that not all BH-like methods are admissible, and a criterion for admissibility of BH-like methods would be of interest.

The class of BH-like methods encompasses all uniform BH improvements valid under PRDS in the literature, and even all uniform improvements under independence the author is aware of. It is unclear whether any other improvements of BH are possible, and whether the class of BH-like methods itself is admissible. Certainly, the BH-like methods are far less elegant than BH itself, and we might hope that other, more elegant uniform improvements might exist. 

Given the completeness of e-Closure for FDR control, any further improvements or alternative methods should benefit from consideration of the e-Closure Principle. The construction of the closed BH method in this paper is, first and foremost, a demonstration of the versatility and power of e-Closure as the central design principle for multiple testing methods.

\appendix

\section{Deferred proofs} \label{sec: proofs}

This section contains the formal proofs of the lemmas in the paper, with some additional helper lemmas. 

\begin{lemma} \label{lem: helper bks}
$b_{k,s} < m$ implies $s < m-1$ and $k+s > m+1$.
\end{lemma}

\begin{proof}
If $b_{k,s} < m$ then $r_{k-1,s} < (k+s-m)(m-s)$. This implies by \eqref{eq: ass r greater k} that
\[
0 \leq r_{k-1,s} - k + 1  < (k+s-m-1)(m-s) + (m-s) - k +1= (k+s-m-1)(m-s-1).
\]
Either both terms are strictly positive or both are strictly negative. Both negative implies $s=m$ and hence $k < 1$, which is inconsistent with the premise $k>0$. So both are strictly positive and the result follows.
\end{proof}

\bkviablefirst*

\begin{proof}
If $b_{k,s} = m$ the result is immediate from \eqref{eq: ass r greater k}. Suppose $b_{k,s} < m$. Then since $r_{k-1,s} \geq k-1 > m-s$ by \eqref{eq: ass r greater k} and \Cref{lem: helper bks}, we have 
\[
(k+s-m-1)(m-s) \geq (k+s-m)(m-s) - r_{k-1,s}.
\]
Since $b_{k,s} < m$ implies $r_{k-1,s} < (k+s-m)(m-s)$, we get 
\[
b_{k,s} = \frac{(m-s)(k+s-m-1)}{(k+s-m)(m-s) - r_{k-1,s}}r_{k-1,s} \geq r_{k-1,s}.
\]
\end{proof}

\bkviablesecond*

\begin{proof}
If $b_{k,s}=m$, the result is immediate. Suppose $b_{k,s}<m$. Note that $b_{k,s}$ is non-decreasing in $r_{k-1,s}$ since the numerator increases in $r_{k-1,s}$ and the denominator is positive and decreasing in $r_{k-1,s}$. Since $r_{k-1,s} \geq k-1$ by \eqref{eq: ass r greater k}, we get, using that $k+s-m-1>0$ and $m-s-1>0$ due to \Cref{lem: helper bks}, 
\begin{align*}
b_{k,s} &\geq \frac{(m-s)(k+s-m-1)(k-1)}{(k+s-m)(m-s) - k+1} = 
\frac{(m-s)(k+s-m-1)(k-1)}{(k+s-m)(m-s-1) + (k+s-m) - k+1} \\ &= 
\frac{(m-s)(k-1)}{m-s-1} = k - 1 + \frac{k-1}{m-s-1} \geq k.
\end{align*}
\end{proof}

\ksmonotone*

\begin{proof}
Let $\p'\geq \p$ coordinatewise, and let $\r_{k,s}'$, $\v_{k,s}'$, $\r'$ and $\k_s'$ be the values of $\r_{k,s}$, $\v_{k,s}$, $\r$ and $\k_s$ calculated using $\p'$. Note that $\p'\geq \p$ coordinatewise implies also that $\p'_{(k)} \geq \p_{(k)}$. We have $\r_{k,s}' \leq \r_{k,s}$ by definition since $a_{k,s}$ does not depend on $\p$. We have $\v_{k,s}' \leq \v_{k,s}$ by the order-preservation of $\max$. $\r'\leq \r$ follows since $\v_{k,s}' \geq k$ implies $\v_{k,s} \geq k$. 
Finally, since both $\{0\leq k \leq \r'\} \subseteq \{0 \leq k \leq \r\}$ and $\{\p_{(k)}' \leq a_{k,s}\} \subseteq \{\p_{(k)} \leq a_{k,s}\}$, and since the maximum over a smaller set is never larger, we have $\k_s' \leq \k_s$.
\end{proof}

\aksmonotone*

\begin{proof}
We need to show that $a_{k,s} \geq a_{k-1,s}$. Consider the case $k\leq m-s+1$ first. We have immediately
\[
a_{k,s} = \frac{k}{s}\alpha \geq \frac{k-1}{s}\alpha = a_{k-1,s}.
\]
Next, consider the case $k > m-s+1$. Write $d=m-s$; $c = k-d = k+s-m$; $r' = r_{k,s}$ and $r= r_{k-1,s}$. We have 
\[
a_{k,s} = \frac \alpha s \bigg(c + \frac{cd}{r'-d}\bigg),
\]
both when $k=m-s+1$ and when $k > m-s+1$. Therefore,
\[
\frac s \alpha (a_{k,s} - a_{k-1,s}) = 1 + \frac{cd}{r'-d} - \frac{d(c-1)}{r-d},
\]
which means we need to show that
\[
\frac{c}{r'-d} - \frac{c-1}{r-d} \geq -\frac{1}{d}.
\]

By \eqref{eq: ass r_k+1 upper bound} we have
\begin{align*}
r' \leq \frac{(c-1)dr}{cd - r} 
= \frac{cdr - dr}{cd - r} 
= \frac{cd^2 - dr + cdr - cd^2}{cd - r}  
= \frac{d(cd-r) + cd(r - d)}{cd - r} 
= d + \frac{cd(r - d)}{cd - r}.
\end{align*}
Therefore, since $c = k+s-m > 0$ and $r' - d \geq k+s-m >0$ by \eqref{eq: ass r greater k}, we have 
\[
\frac{c}{r'-d} \geq \frac{cd - r}{d(r - d)} 
= \frac{(c-1)d - (r-d)}{d(r - d)} \\
= \frac{c-1}{r - d} - \frac 1d.
\]
Note that we have equality if $r' = b_{k,s} < m$.
\end{proof}

\minimale*

\begin{proof}
We distinguish between the cases $k \leq m-s+1$ and $k > m-s+1$. Write $t=|[k] \cap S|$, and $r=r_{k,s}$.

Let $k \leq m-s+1$. Then $r = k$, so we only need to consider $j=k$.
\[
\frac1s \sum_{i \in S} e_i \geq \frac ts \frac{s}{k\alpha} = \frac{t}{k\alpha} = \frac{S \cap [j]}{j\alpha} = \frac{|R \cap S|}{|R|\alpha}.
\]

Let $k > m-s+1$. We have $t \geq s - (m-k) = k+s-m$, because at most $m-k$ elements of $S$ fit into $[m]\setminus [k]$. We have
\[
\frac1{|S|} \sum_{i \in S} e_i \geq \frac ts  \frac{(r+s-m)s}{(k+s-m)r\alpha} = \frac{(r+s-m)t}{(k+s-m)r\alpha}.
\]
The set $S \cap [j]$ contains at most $t$ elements from $[k]$ and at most $j-k$ elements from $[j]\setminus [k]$. Therefore, 
\[
\frac{|R \cap S|}{|R|\alpha} \leq \frac{t+j-k}{j\alpha}.
\]
It suffices to show that, for $k \leq j \leq r$, we have  
\[
\frac{(r+s-m)t}{(k+s-m)r\alpha} \geq \frac{t+j-k}{j\alpha},
\]
or equivalently (dropping $\alpha$) that
\[
h(j) = (r+s-m)tj - (k+s-m)(t+j-k)r \geq 0.
\]
This function is linear in $j$, so we only need to check the endpoints. We have 
\[
h(k) = (r+s-m)tk - (k+s-m)tr = t(m-s)(r-k) \geq 0,
\]
since $r \geq k$, and
\begin{eqnarray*}
h(r) &=& (r+s-m)tr - (k+s-m)(t+r-k)r \\ &=& 
tr(r-k) - (k+s-m)(r-k)r \\ &=& r(r-k)(t-k-s+m) \geq 0,
\end{eqnarray*}
since $r \geq k$ and $t \geq k+s-m$.
\end{proof}

\eSOK*

\begin{proof}
Choose any $S \in 2^{[m]}$ and let $s = |S|$. If $s=0$ or $\r=0$ there is nothing to prove, so assume $s>0$ and $\r>0$. 

By definition of $\r$, there exists $k' \leq \r$ such that $\r_{k',s} \geq \r$. Since $r_{k,s}$ is non-decreasing in $k$, we can take this $k'$ as the largest index $\leq \r$ such that $\r_{k',s} > 0$, which is $k' = \k_s$. Therefore $\k_s \leq \r \leq r_{\k_s,s}$.

By definition of $\k_s$, we have $\p_{(\k_s)} \leq a_{\k_s, s} = \a_s$. Define $\e_{i,s} = 1\{\p_{(i)} \leq a_{\k_s,s}\}/a_{\k_s,s}$. Then $\e_{i,s} \geq 1/a_{\k_s,s}$ for $i \leq \k_s$, so by \Cref{lem: minimal e}, taking the permutation of the indices into account, we have $\e_S \geq |\I_j \cap S|/j\alpha$ for all $\k_s \leq j \leq r_{\k_s, s}$. Since $\k_s \leq \r \leq r_{\k_s,s}$, this holds in particular for $j=\r$.
\end{proof}

\lemMABH*

\begin{proof}
By definition of $\b'$ we have (1.) $\p_{(i)} \leq i\alpha/m$ for at least one $i \leq \b'$; (2.) $\p_{(\b')} \leq \b'\alpha/(m-1)$;  and (3.) $\p_{(j)} > j\alpha/(m-1)$ for all $j > \b'$.

In this situation we have, for $s=m-1$, $\p_{(\b')} \leq a_{\b',m-1}$, so $\r_{\b',m-1}=\b'$; and $\p_{(j)} > a_{j,m-1}$ for $j> \b'$, so $\r_{j,m-1}=0$ for $j > \b'$. It follows that $\check\r \leq \b'$.   

For $s=m$, $\p_{(i)} \leq a_{i,m}$, so $\r_{i,m}=m$; For $s \leq m-1$ we have $\p_{(\b')} \leq a_{\b', m-1} = \b'\alpha/(m-1) \leq \b'\alpha/s = a_{\b',s}$, so $\r_{\b',s} = \b'$. Therefore $\check\r \geq \b'$. 
\end{proof}

\begin{lemma} \label{lem: aks max ma}
    $a_{k,s} \leq m\alpha$.
\end{lemma}

\begin{proof}
By \Cref{lem: aks nondecreasing}, $a_{k,s}$ is maximized when $k=m$, which implies $r_{k,s} = m$, so that $a_{m,s} = m\alpha/s$. This in turn is maximized when $s=1$.
\end{proof}

\MABHadmissible*

\begin{proof}
Suppose $r_{k,s}$ is such that $\r \geq \check\r$. We wil first show that this implies that $r_{k,m-1} = k$ for all $k$, then that $r_{k,m-1} = k$ for all $k$ in turn implies that $\r =\check\r$. Choose $\alpha<1/m$.

Suppose for some $k$, $r_{k,m-1}>k$; we must have $k>2$. Define $\p_1=\alpha/m$; $\p_2=\ldots=\p_k = k\alpha/(m-1)$, and $\p_{k+1} =\ldots = \p_m=1$. Then $\check \r = k$. We have
\[
\p_{(k)} = \frac{k\alpha}m > \frac{(k+s-m)r_{k,s}\alpha}{(r_{k,s} +s-m)s} = a_{k,s}, 
\]
since $a$ is increasing in $r_{k,s}$. For all $j> k$, we have
\[
\p_{(j)} = 1 > a_{j,s},
\]
by \Cref{lem: aks max ma}. It follows that for all $j \geq k$, we have $\v_{j,m-1} \leq r_{m-1,k-1} = k-1$, so $\r \leq k-1 < \check \r$. So $\r$ cannot uniformly improve $\check\r$.

Next, suppose $r_{k,m-1}=k$ for all $k$. Then by definition of $\check\r$ we have $\p_{(\check\r)} \leq \check\r\alpha/(m-1) = a_{k,s}$ and for all $j > \check\r$, $\p_{(j)} > j\alpha/(m-1) = a_{j,s}$. Therefore, $\r_{k,m-1} = k$ when $k \leq \check\r$ and $\r_{k,m-1} =0$ when $k > \check\r$. It follows that $\r \leq \check\r$.
\end{proof}

\criterionBH*

\begin{proof}
Let $k+s > m+1$, $s<m$ and $k<m$. We have, since $a_{k,s}$ is strictly decreasing in $r_{k,s}$, that $r_{k,s} \leq ks/(m-k)$ if and only if
\[
a_{k,s} = \frac{(k+s-m)r_{k,s}\alpha}{(r_{k,s}+s-m)s} 
\geq \frac{(k+s-m)\frac{ks}{m-k}\alpha}{(\frac{ks}{m-k}+s-m)s} 
= \frac{(k+s-m)ks\alpha}{(ks + (m-k)(s-m))s} 
= \frac{k\alpha}{m}.
\]
If $k+s > m+1$, or $s<m$, or $k<m$, we have $a_{k,s} = k\alpha/s \geq k\alpha/m$ always. Therefore, $a_{k,s}  \geq k\alpha/m$ if and only if $r_{k,s} \leq ks/(m-k)$ whenever $k+s > m+1$, $s<m$ and $k<m$. 

We will now show that $a_{k,s}  \geq k\alpha/m$ if and only if $\r \geq \b$ for all $\p$ and $\alpha$.

Suppose $a_{k,s} \geq k\alpha/m$. Then $\p_{(\b)} \leq \b\alpha/m$ implies $\p_{(\b)} \leq a_{\b, s}$ for all $s$, so $\r \geq \b$. 

Suppose $a_{k,s} < k\alpha/m$ for some $k,s$. Choose $\alpha < 1/m$. Define $\p_i = k\alpha/m$ for $i \leq k$ and $\p_i=1$ for $i > k$, so $\b=k>0$. Then we have, for $i \leq k$, by \Cref{lem: aks nondecreasing}, $p_{(i)} = k\alpha/m > a_{k,s} \geq a_{i,s}$. For $i > k$, we have $\p_{(i)} = 1 > m\alpha \geq a_{i,s}$ by \Cref{lem: aks max ma}. Therefore, $\r_{i,s}=0$ for all $i$, so $\r = 0 < \b$.
\end{proof}

\lemsqrtm*

\begin{proof}
Let $\bar a_{k,s}$ and $\bar \r_{k,s}$ refer to the corresponding quantities for $\bar\r$. Let $k+s > m+1$ and $k < \sqrt{m}$. Then
\[
\frac{ks}{m-k} = \frac{k(m-k) + k(k+s-m)}{m-k} 
= k+ \frac{k^2 + s- m}{m-k}  < k + \frac{s}{m-k} \leq k+ 1. 
\]
Therefore, $\bar r_{k,s} = k$ for $k < \sqrt{m}$.

Suppose $\check\r < \sqrt{m}$. Then by \Cref{lem: MABH} there exists $i < \sqrt{m}$ such that $\p_{(i)} \leq i\alpha/m$ and $i \leq j \leq \sqrt{m}$ such that $\p_{(j)} \leq i\alpha/(m-1)$. Further, for all $k>j$, we have $\p_{(k)} > k\alpha/(m-1)$.

In this situation we have, for $s=m-1$, $\p_{(\check\r)} \leq \check \r\alpha/(m-1) = \bar a_{\check\r,m-1}$, so $\bar \r_{\check\r,m-1}=\check\r$; and $\p_{(k)} > k \alpha/(m-1) = \bar a_{k,m-1}$ for $k> \check\r$, so $\bar\r_{k,m-1}=0$ for $k > \check\r$. It follows that $\bar\r \leq \check\r$.   

For $s=m$, $\p_{(i)} \leq i\alpha/m = \bar a_{i,m}$, so $\bar\r_{i,m}=m$; For $s \leq m-1$ we have $\p_{(\check r)} \leq \check\r\alpha/(m-1) \leq \check\r\alpha/s = \bar a_{\check\r,s}$, so $\bar\r_{\check\r,s} = \check\r$. Therefore $\bar\r \geq \check\r$. 
\end{proof}

\begin{lemma} \label{lem: critical s}
If $\r>0$, there exists $s<m$ such that $\r_{\k_s,s} = \r$ and, if $\k_s < m$, that $p_{(\k_s+1)} > a_{\k_s,s}$.
\end{lemma}

\begin{proof}
If $\r=m$, the statement is immediate. Let $\r<m$. By definition of $\r$, there exists a witness $s$ such that $\v_{\r+1,s} < \r+1$, while $\v_{\r,s} \geq \r$ (since this holds for all $s$). 

We claim that any such $s$ satisfies $s<m$. Indeed, for $k \geq 1$ we have $r_{k,m}=m$ by \eqref{eq: ass r initialization} and \eqref{eq: ass r_k+1 upper bound}, so $\r_{k,m} \in \{0,m\}$ and also $\v_{k,m} \in \{0,m\}$. Since $\v_{\r,m} \geq \r \geq 1$, we have $\v_{\r, m}=m$. Consequently, $\v_{\r+1,m} = m \geq \r+1$, which violates the condition $\v_{\r+1,s} < \r+1$ we established above for the witness $s$.

For the witness $s$, we have $\r \leq \r_{\k_s,s} < \r+1$, so $\r_{\k_s,s} = \r$. If $\k_s < \r$, we have $\r_{\k_s+1,s} =0$ since $\k_s$ is maximal; if $\k_s=\r$ we have $\r_{\k_s+1,s} =0$ since otherwise $\r+1 \leq \r_{\k_s+1,s} = \v_{\r+1,s} < \r+1$. So $\p_{(\k_s+1)} > a_{\k_s+1,s} \geq a_{k_s,s}$.
\end{proof}

\begin{lemma} \label{lem: curly R size}
If $|R| \neq \r$, then $R \notin \boldsymbol{\mathcal{R}}$.
\end{lemma}

\begin{proof}
Assume $R \in \boldsymbol{\mathcal{R}}$ and $|R|\neq \r$. We will arrive at a contradiction by finding an $S$ such that $\e_S < |S \cap R|/|R|\alpha$.

Suppose $|R| < \r$. Choose any $i \in R$. We have, for $S=\{i\}$,
\[
\e_{\{i\}} \leq \frac1{\a_1} = \frac1{a_{\k_1,1}} = \frac1{a_{\r,1}} =  \frac1{\r\alpha} < \frac1{|R|\alpha} = \frac{|\{i\} \cap R|}{|R|\alpha}.
\]

Suppose $|R|> \r$. Let $s<m$ be as constructed in \Cref{lem: critical s}. 

Suppose $\k_s+s \leq m+1$, then $\r = \k_s$. Choose any $i \in R\setminus \I_\r$. Then, for $S=\{i\}$, 
\[
\e_{\{i\}} = 0 < \frac{1}{|R|\alpha} = \frac{|R \cap \{i\}|}{|R|\alpha}.
\]

Suppose $\k_s+s > m+1$. Since all $\e_S$ are decreasing in $\p$, and $\r\geq \k_s$, we may assume that $R \supset \I_{\k_s}$. By \Cref{lem: critical s}, 
\[
\e_{i,s} = \frac{1\{\p_i \leq a_{\k_s,s}\}}{a_{\k_s,s}}
\]
is non-zero only for $i \in \I_{\k_s}$. Choose any $S \supset [m]\setminus \I_{\k_s}$ with $|S|=s$, which is possible since $s > m-\k_s$. Then $|S \cap \I_{\k_s}| = \k_s+s-m$, so
\[
\e_S = \frac{|S \cap \I_{\k_s}|}{sa_{\k_s,s}} = \frac{|S \cap \I_{\k_s}|}{s} \frac{(\r + s-m)s}{(\k_s+s-m)\r\alpha} = \frac{\r + s-m}{\r\alpha}.
\]
By construction of $S$, we have, since $R \supset \I_{\k_s}$,
\[
|S \cap R| = |S \cap \I_{\k_s}| + |R| - \k_s = |R|+s-m.
\]
Since $(r+s-m)/r$ is increasing in $r$ because $s<m$, we have
\[
\e_S = \frac{\r + s-m}{\r\alpha} < \frac{|R| + s-m}{|R|\alpha} = \frac{|S \cap R|}{|R|\alpha}.
\]
\end{proof}

\begin{lemma} \label{lem: sufficient curly R}
Let $|R|=\r>0$. Then $R \in \boldsymbol{\mathcal{R}}$ if and only if $|\{i\in R\colon \p_i \leq a_{\k_s,s}\}| \geq \k_s$ , for all $s$.
\end{lemma}

\begin{proof}
Suppose $|\{i\in R\colon \p_i \leq a_{\k_s,s}\}| \geq \k_s$ , for all $s$. Let $\pi$ be a permutation that sorts $R$ before $[m]\setminus R$, Then define $\e_{i,s} = \frac{1\{\p_{\pi_i} \leq a_{\k_s,s}\}}{a_{\k_s,s}}$. By sorting the elements of $R$ first, and by the assumption, we fulfill the condition of \Cref{lem: minimal e} for every $s$. Therefore,
\[
\e_S \geq \frac{|R'\cap S|}{|R'|\alpha},
\]
for $S$ with $|S|=s$ and all $R' = \{\pi_1, \ldots, \pi_j\}$ for $\k_s \leq j \leq r_{\k_s,s}$. Since $\k_s \leq \r \leq r_{\k_s,s}$, this holds in particular for $R$. Therefore, $R \in \boldsymbol{\mathcal{R}}$.

Let $|\{i\in R\colon \p_i \leq a_{k_s,s}\}| < \k_s\leq \r$ for at least one $s$, which implies $s<m$. Define $A=\{i\in R\colon \p_i \leq a_{k_s,s}\}$, so $|R \setminus A|>0$. Since $K=\{i\in [m]\colon \p_i \leq a_{k_s,s}\}$ has $|K|\geq \k_s$ by definition of $k_s$, we have $|K\setminus R| \geq 1$. Choose $S$ with $|S|=s$ such that it exhausts first $R \setminus A$, then $[m]\setminus K$, then $A$ and finally $K \setminus R$. For this $S$ we have, if $\k_s+s > m$, that $|S \cap K| = \k_s+s-m$, and $|S\cap R| > \r+s-m$, so
\[
\e_S = \frac1s \sum_{i \in S} \frac{1}{a_{\k_s}}1\{\p_i \leq a_{\k_s,s}\} = \frac{|K\cap S|(\r+s-m)}{(\k_s+s-m)\r\alpha} < \frac{|R \cap S|}{|R|\alpha}.
\]
If $\k_s+s \leq m$, we have $|K\cap S| < |S\cap R|$. Therefore,
\[
\e_S = \frac1s \sum_{i \in S} \frac{1}{a_{\k_s}}1\{\p_i \leq a_{\k_s,s}\} = \frac{|S \cap K|}{\r\alpha} < \frac{|S\cap R|}{|R|\alpha}.  
\]
It follows that $R \notin \boldsymbol{\mathcal{R}}$.
\end{proof}

\hardlysimultaneous*

\begin{proof}
Combine \Cref{lem: curly R size} and \Cref{lem: sufficient curly R}.
\end{proof}

\section{Algorithm} \label{sec: alg adjusted}

Algorithm \ref{alg:cBH:adjust} gives the pseudo-code for the calculation of adjusted p-values for closed BH.

\begin{algorithm}[H]
\caption{Adjusted p-values for the general BH-like method}\label{alg:cBH:adjust}
\KwIn{$\p_1, \ldots, \p_m$: p-values; $\mathtt{next.r}$: horizon function}
\KwOut{$\q_1, \ldots, \q_m$: adjusted p-values}
 
Sort and let $\p_{(1)} \leq \cdots \leq \p_{(m)}$\;
Initialize $\q_{(k)} \gets 0$ for $k = 1, \ldots, m$\;
 
\BlankLine
\For{$s = m, m-1, \ldots, 1$}{
 
  Initialize empty deque $\mathcal{D}$\tcp*[r]{stores $(k', \; \p_{(k')}/\tilde a_{k',s})$ pairs; min at front}
  $l \gets m - s + 1$\tcp*[r]{left edge of qualifying window}
 
  \BlankLine
  \For{$k = m-s+1, \ldots, m$}{
    Compute $r_{k,s}$ and $a_{k,s}$ as in Algorithm~\ref{alg:cBH}\tcp*[r]{$r$ updated incrementally via $\mathtt{next.r}$}
    $\rho_k \gets \p_{(k)} \,/\, \tilde a_{k,s}$\;
 
    \BlankLine
    \While{$\mathcal{D} \neq \emptyset$ \textbf{and} $\mathrm{back}(\mathcal{D}).\rho \geq \rho_k$}{
      Remove back of $\mathcal{D}$\tcp*[r]{maintain deque monotonicity}
    }
    Push $(k, \rho_k)$ to back of $\mathcal{D}$\;
 
    \BlankLine
    \While{$l \leq k$ \textbf{and} $r_{l,s} < k$}{
      $l \gets l + 1$\tcp*[r]{advance window: need $r_{k',s} \geq k$}
    }
    \While{$\mathcal{D} \neq \emptyset$ \textbf{and} $\mathrm{front}(\mathcal{D}).k' < l$}{
      Remove front of $\mathcal{D}$\tcp*[r]{evict entries outside window}
    }
 
    \BlankLine
    \eIf{$\mathcal{D} = \emptyset$}{
      $\q_{(j)} \gets 1$ for all $j = k, \ldots, m$\tcp*[r]{no qualifying $k'$; window only shrinks}
      \textbf{break}\;
    }{
      $\q_{(k)} \gets \max\!\bigl(\q_{(k)}, \;\mathrm{front}(\mathcal{D}).\rho\bigr)$\tcp*[r]{sliding window minimum}
    }
  }
}
 
\BlankLine
\tcp{Post-processing}
\For(\tcp*[f]{function returns max survivor, not individual survivors}){$k = m-1, \ldots, 1$}{
  $\q_{(k)} \gets \min(\q_{(k)}, \; \q_{(k+1)})$\;
}
$\q_{(k)} \gets \max(\q_{(k)}, \; \p_{(k)})$ for $k = 1, \ldots, m$\tcp*[r]{optional: adjusted $\geq$ raw}
Restore original ordering\;
 
\Return{$\q_1, \ldots, \q_m$}\;
\end{algorithm}

\bibliography{BH} 

@article{BenjaminiHochberg1995,
  author    = {Benjamini, Yoav and Hochberg, Yosef},
  title     = {Controlling the False Discovery Rate: A Practical and Powerful Approach to Multiple Testing},
  journal   = {Journal of the Royal Statistical Society: Series B (Methodological)},
  volume    = {57},
  number    = {1},
  pages     = {289--300},
  year      = {1995},
  doi       = {10.1111/j.2517-6161.1995.tb02031.x}
}

@article{BenjaminiHochberg2000,
  author    = {Benjamini, Yoav and Hochberg, Yosef},
  title     = {On the Adaptive Control of the False Discovery Rate in Multiple Testing with Independent Statistics},
  journal   = {Journal of Educational and Behavioral Statistics},
  volume    = {25},
  number    = {1},
  pages     = {60--83},
  year      = {2000},
  doi       = {10.3102/10769986025001060}
}

@article{BenjaminiYekutieli2001,
  title={The control of the false discovery rate in multiple testing under dependency},
  author={Benjamini, Yoav and Yekutieli, Daniel},
  journal={Annals of statistics},
  pages={1165--1188},
  year={2001},
  publisher={JSTOR}
}

@article{Storey2002,
  author    = {Storey, John D.},
  title     = {A Direct Approach to False Discovery Rates},
  journal   = {Journal of the Royal Statistical Society: Series B (Statistical Methodology)},
  volume    = {64},
  number    = {3},
  pages     = {479--498},
  year      = {2002},
  doi       = {10.1111/1467-9868.00346}
}

@article{StoreyTaylorSiegmund2004,
  author    = {Storey, John D. and Taylor, Jonathan E. and Siegmund, David},
  title     = {Strong Control, Conservative Point Estimation and Simultaneous Conservative Consistency of False Discovery Rates: A Unified Approach},
  journal   = {Journal of the Royal Statistical Society: Series B (Statistical Methodology)},
  volume    = {66},
  number    = {1},
  pages     = {187--205},
  year      = {2004},
  doi       = {10.1111/j.1467-9868.2004.00439.x}
}

@article{BenjaminiKriegerYekutieli2006,
  author    = {Benjamini, Yoav and Krieger, Abba M. and Yekutieli, Daniel},
  title     = {Adaptive Linear Step-Up Procedures that Control the False Discovery Rate},
  journal   = {Biometrika},
  volume    = {93},
  number    = {3},
  pages     = {491--507},
  year      = {2006},
  doi       = {10.1093/biomet/93.3.491}
}

@article{BlanchardRoquain2008,
  author    = {Blanchard, Gilles and Roquain, Etienne},
  title     = {Two Simple Sufficient Conditions for {FDR} Control},
  journal   = {Electronic Journal of Statistics},
  volume    = {2},
  pages     = {963--992},
  year      = {2008},
  doi       = {10.1214/08-EJS180}
}

@article{BlanchardRoquain2009,
  author    = {Blanchard, Gilles and Roquain, Etienne},
  title     = {Adaptive {FDR} Control under Independence and Dependence},
  journal   = {Journal of Machine Learning Research},
  volume    = {10},
  pages     = {2837--2871},
  year      = {2009}
}

@article{Sarkar2008,
  author    = {Sarkar, Sanat K.},
  title     = {Two-Stage Stepup Procedures Controlling {FDR}},
  journal   = {Journal of Statistical Planning and Inference},
  volume    = {138},
  number    = {4},
  pages     = {1072--1084},
  year      = {2008},
  doi       = {10.1016/j.jspi.2007.03.058}
}

@article{GavrilovBenjaminiSarkar2009,
  author    = {Gavrilov, Yulia and Benjamini, Yoav and Sarkar, Sanat K.},
  title     = {An Adaptive Step-Down Procedure with Proven {FDR} Control under Independence},
  journal   = {The Annals of Statistics},
  volume    = {37},
  number    = {2},
  pages     = {619--629},
  year      = {2009},
  doi       = {10.1214/07-AOS586}
}

@article{FinnerDickhausRoters2009,
  author    = {Finner, Helmut and Dickhaus, Thorsten and Roters, Markus},
  title     = {On the False Discovery Rate and an Asymptotically Optimal Rejection Curve},
  journal   = {The Annals of Statistics},
  volume    = {37},
  number    = {2},
  pages     = {596--618},
  year      = {2009},
  doi       = {10.1214/07-AOS569}
}

@article{HeesenJanssen2016,
  author    = {Heesen, Philipp and Janssen, Arnold},
  title     = {Dynamic Adaptive Multiple Tests with Finite Sample {FDR} Control},
  journal   = {Journal of Statistical Planning and Inference},
  volume    = {168},
  pages     = {38--51},
  year      = {2016},
  doi       = {10.1016/j.jspi.2015.06.007}
}

@article{MacDonaldLiangJanssen2019,
  author    = {MacDonald, Peter and Liang, Kun and Janssen, Arnold},
  title     = {Dynamic Adaptive Procedures that Control the False Discovery Rate},
  journal   = {Electronic Journal of Statistics},
  volume    = {13},
  number    = {2},
  pages     = {3009--3024},
  year      = {2019},
  doi       = {10.1214/19-EJS1589}
}

@article{SolariGoeman2017,
  author    = {Solari, Aldo and Goeman, Jelle J.},
  title     = {Minimally Adaptive {BH}: A Tiny but Uniform Improvement of the Procedure of {Benjamini} and {Hochberg}},
  journal   = {Biometrical Journal},
  volume    = {59},
  number    = {4},
  pages     = {776--780},
  year      = {2017},
  doi       = {10.1002/bimj.201500253}
}

@article{Gao2023,
  author    = {Gao, Zijun},
  title     = {Adaptive {Storey's} Null Proportion Estimator},
  journal   = {arXiv preprint arXiv:2310.06357},
  year      = {2023},
  note      = {arXiv:2310.06357}
}

@article{GaoRoquain2025,
  author    = {Gao, Zijun and Roquain, Etienne},
  title     = {On Min-{Storey} Estimators for Multiple Testing and Conformal Novelty Detection},
  journal   = {arXiv preprint arXiv:2603.17984},
  year      = {2025},
  note      = {arXiv:2603.17984}
}

@article{IgnatiadisWangRamdas2026,
  author    = {Ignatiadis, Nikolaos and Wang, Ruodu and Ramdas, Aaditya},
  title     = {Tiny but Uniform Improvements of Adaptive {BH} Procedures via Compound E-Values},
  journal   = {arXiv preprint arXiv:2603.21424},
  year      = {2026},
  note      = {arXiv:2603.21424}
}

@article{XuSolariFischerDeHeideRamdasGoeman2025,
  title={Bringing closure to false discovery rate control: A general principle for multiple testing},
  author={Xu, Ziyu and Solari, Aldo and Fischer, Lasse and de Heide, Rianne and Ramdas, Aaditya and Goeman, Jelle},
  journal={arXiv preprint arXiv:2509.02517},
  year={2025}
}

@article{WangRamdas2022,
  title={False discovery rate control with e-values},
  author={Wang, Ruodu and Ramdas, Aaditya},
  journal={Journal of the Royal Statistical Society Series B: Statistical Methodology},
  volume={84},
  number={3},
  pages={822--852},
  year={2022},
  publisher={Oxford University Press}
}

@article{Su2018,
  title={The FDR-linking theorem},
  author={Su, Weijie J},
  journal={arXiv preprint arXiv:1812.08965},
  year={2018}
}

@article{GoemanSolari2014,
  title={Multiple hypothesis testing in genomics},
  author={Goeman, Jelle J and Solari, Aldo},
  journal={Statistics in medicine},
  volume={33},
  number={11},
  pages={1946--1978},
  year={2014},
  publisher={Wiley Online Library}
}

@book{EfronHastie2021,
  title={Computer age statistical inference, student edition: algorithms, evidence, and data science},
  author={Efron, Bradley and Hastie, Trevor},
  volume={6},
  year={2021},
  publisher={Cambridge University Press}
}

@article{FithianLei2022,
  title={Conditional calibration for false discovery rate control under dependence},
  author={Fithian, William and Lei, Lihua},
  journal={The Annals of Statistics},
  volume={50},
  number={6},
  pages={3091--3118},
  year={2022},
  publisher={Institute of Mathematical Statistics}
}

@article{LiGoeman2026,
  title={On the error control of invariant causal prediction},
  author={Li, Jinzhou and Goeman, Jelle J},
  journal={arXiv preprint arXiv:2401.03834},
  year={2026}
}

@article{FinnerRoters2001,
  title={On the false discovery rate and expected type I errors},
  author={Finner, Helmut and Roters, Markus},
  journal={Biometrical Journal},
  volume={43},
  number={8},
  pages={985--1005},
  year={2001},
  publisher={Wiley Online Library}
}

@article{KatsevichSabattiBogomolov2023,
  title={Filtering the rejection set while preserving false discovery rate control},
  author={Katsevich, Eugene and Sabatti, Chiara and Bogomolov, Marina},
  journal={Journal of the American Statistical Association},
  volume={118},
  number={541},
  pages={165--176},
  year={2023},
  publisher={Taylor \& Francis}
}

@article{EbrahimpoorGoeman2021,
  title={Inflated false discovery rate due to volcano plots: problem and solutions},
  author={Ebrahimpoor, Mitra and Goeman, Jelle J},
  journal={Briefings in bioinformatics},
  volume={22},
  number={5},
  pages={bbab053},
  year={2021},
  publisher={Oxford University Press}
}

@article{MarcusPeritzGabriel1976,
  title={On closed testing procedures with special reference to ordered analysis of variance},
  author={Marcus, Ruth and Eric, Peritz and Gabriel, K Ruben},
  journal={Biometrika},
  volume={63},
  number={3},
  pages={655--660},
  year={1976},
  publisher={Oxford University Press}
}

@article{VovkWang2021,
  author  = {Vovk, Vladimir and Wang, Ruodu},
  title   = {{E}-values: {C}alibration, combination, and applications},
  journal = {The Annals of Statistics},
  year    = {2021},
  volume  = {49},
  number  = {3},
  pages   = {1736--1754},
  doi     = {10.1214/20-AOS2020}
}

@article{GrunwaldDeHeideKoolen2024,
  author  = {Gr\"{u}nwald, Peter and de Heide, Rianne and Koolen, Wouter},
  title   = {Safe testing},
  journal = {Journal of the Royal Statistical Society Series B: Statistical Methodology},
  year    = {2024},
  volume  = {86},
  number  = {5},
  pages   = {1091--1128},
  doi     = {10.1093/jrsssb/qkae011}
}

@article{RamdasWang2025,
  author    = {Ramdas, Aaditya and Wang, Ruodu},
  title     = {Hypothesis Testing with {E}-values},
  journal   = {Foundations and Trends in Statistics},
  year      = {2025},
  volume    = {1},
  number    = {1-2},
  pages     = {1--390},
  doi       = {10.1561/3600000002}
}

@article{Sarkar2008PRDS,
  title={On methods controlling the false discovery rate},
  author={Sarkar, Sanat K},
  journal={Sankhy{\=a}: The Indian Journal of Statistics, Series A (2008-)},
  pages={135--168},
  year={2008},
  publisher={JSTOR}
}

\end{document}